\documentclass[journal]{IEEEtran}
\usepackage{helvet}         % selects Helvetica as sans-serif font
\usepackage{courier}        % selects Courier as typewriter font
\usepackage{type1cm}        % activate if the above 3 fonts are
\usepackage{makeidx}         % allows index generation
\usepackage{graphicx}        % standard LaTeX graphics tool

\usepackage{multicol}        % used for the two-column index
\usepackage[bottom]{footmisc}% places footnotes at page bottom
\usepackage{ifthen}
\usepackage{subfigure}
\usepackage[usenames,dvipsnames]{color}
\usepackage{cite}

\makeindex            
\usepackage{graphics} % for pdf, bitmapped graphics files
\usepackage{graphicx} % for pdf, bitmapped graphics files
\usepackage{epsfig} % for postscript graphics files

\usepackage{times} % assumes new font selection scheme installed
\usepackage{mathtools}  % loads amsmath
\usepackage{amssymb}  % assumes amsmath package installed
\usepackage{amsthm}
\usepackage{amsfonts}
\usepackage{dsfont}
\DeclareGraphicsExtensions{.pdf,.png,.jpg,.eps}

\usepackage{subfig}
\usepackage{verbatim}
\usepackage{comment}
\usepackage{algorithm}
\usepackage{algorithmic}
\usepackage{multirow}
\usepackage{enumitem}
\usepackage{caption}
\usepackage{setspace}
\DeclareGraphicsExtensions{.pdf,.png,.jpg,.eps}

\newcommand{\rr}{{\mathbb{R}}}

\newboolean{showcomments}
\setboolean{showcomments}{true}
\newcommand{\wang}[1]{\ifthenelse{\boolean{showcomments}}
	{ \textcolor{red}{(ZW:  #1)}}{}}
\newcommand{\fliu}[1]{\ifthenelse{\boolean{showcomments}}
	{ \textcolor{red}{(FL:  #1)}}{}}
\newcommand{\peng}[1]{\ifthenelse{\boolean{showcomments}}
	{ \textcolor{red}{(PY:  #1)}}{}}

\theoremstyle{definition}
\newtheorem{theorem}{Theorem}
\newtheorem{lemma}[theorem]{Lemma}

\newtheorem{proposition}{Proposition}

\theoremstyle{definition}
\newtheorem{definition}{Definition}
\newtheorem{remark}{Remark}

\newtheorem{example}{Example}
\newtheorem{assumption}{Assumption}

\newtheorem{problem}{Problem}
\begin{document}
	\setstretch{1}	
	
	\title{Augmented Synchronization of Power Systems}
	
	\author{Peng~Yang, ~Feng~Liu, \IEEEmembership{Senior~Member,~IEEE,}	
		~Tao Liu, \IEEEmembership{Member,~IEEE,}
		~David J. Hill, \IEEEmembership{Life~Fellow,~IEEE}	
		
		%		\thanks{This work was supported by the National Natural Science Foundation of China (No. 51677100, U1766206) (Corresponding author: Feng Liu)}
		%		\thanks{P. Yang and F. Liu are with the Department of Electrical Engineering, Tsinghua University, Beijing, 100084, China. (e-mail: lfeng@tsinghua.edu.cn)}	
	}
	
	\maketitle
	
	\begin{abstract}
		Power system transient stability has been translated into a Lyapunov stability problem of the post-disturbance equilibrium for decades. Despite substantial results, conventional theories suffer from the stringent requirement of knowing the post-disturbance equilibrium a priori. In contrast, the wisdom from practice, which certificates stability by only the observation of converging frequencies and voltages, seems to provide an equilibrium-independent approach. Here, we formulate the empirical wisdom by the concept of augmented synchronization and aim to bridge such a theory-practice gap. First, we derive conditions under which the convergence to augmented synchronization implies the convergence to the equilibrium set, laying the first theoretical foundation for the empirical wisdom. Then, we reveal from what initial values the power system can achieve augmented synchronization. Our results open the possibility of an equilibrium-independent power system stability analytic that re-defines the nominal motion as augmented synchronization rather than certain equilibrium. Single-machine examples and the IEEE 9-bus system well verify our results and illustrate promising implications.
	\end{abstract}
	

	\IEEEpeerreviewmaketitle

	\section{Introduction}
	\label{sec:1}
	\IEEEPARstart{T}{ransient} stability underlies functional operations of modern power grids, which usually span thousands of kilometers in open land and always suffer from various types of disturbances. It refers to the ability of a power system, for a given initial operating condition, to regain a (new) state of operating equilibrium after being subjected to a large disturbance such as short-circuit faults and sudden large load changes \cite{definition_stability}. Although this descriptive definition covers the essence of transient stability, its interpretation diverges for theorists and engineering practitioners. 
	%	\peng{Our translation. differ stability and convergence. Motivation, the gap among physical definition, theory, and practice wisdom. Three parties. Properly address why we use convergence instead of stability}
	
	Theoretically, transient stability has been translated into the equilibrium stability problem in the sense of Lyapunov for decades \cite{Kundur_Powersystemstability_1994}. Under this framework, a set of ordinary differential equations (ODEs) or differential algebraic equations (DAEs) are used to describe the post-disturbance dynamics of a power system\cite{definition_stability}. Theorists are interested in whether a post-disturbance equilibrium is Lyapunov asymptotically stable, and if so, what is the region of attraction (RoA), i.e., from which initial point the system solution can converge to this equilibrium. This idea has led to the so-called direct methods that are based on Lyapunov functions or energy functions \cite{pai2012energy,fouad1991power,chiang1987foundations,chiang2011direct}, and can directly assess transient stability without time-consuming simulations. Despite substantial results, such theories have been criticized for the stringent requirement on prior knowledge of the post-disturbance equilibrium. Such a requirement is unrealistic especially when the post-disturbance equilibrium depends on initial points. In some cases, system trajectories may converge to an equilibrium set, but none of the equilibria is asymptotically stable, and hence conventional direct methods fail. Indeed, taking a single stable equilibrium as the subject greatly restricts the capability of conventional transient stability theories. A rudimentary example will be shown in Section \ref{sec:case}, where these challenging issues can arise simply from interactions among subsystems.
	
	Fortunately, power system engineering practitioners have already found an intuitive and effective way to circumvent the aforementioned dilemma. Instead of an equilibrium-dependent stability concept, they often interpret and assess transient stability differently. After a large disturbance, if all frequencies synchronize to around the nominal value (50 or 60 Hz), and all voltages converge to steady-state values within a certain safety region, then the power system will be regarded to successfully regain an operating point, and hence is transiently stable. This practical criterion demands no information of the post-disturbance equilibrium but only the observation of converging frequencies and voltages. This feature is appealing in practice as it is impossible to monitor all state variables while the measurement of frequencies and voltages are often easy to obtain \cite{monticelli2000electric}. Although this empirical wisdom works well, one intriguing and important question remains: does the convergence of only frequencies and voltages guarantee the convergence of all states? After all, it is the latter, not the former, that the transient stability actually concerns.
	
	A mismatch between ``demand" and ``supply" of theories exists as well. Although \textit{equilibrium-independent} transient stability analytics has been advocated in the power system community for decades, which, to the best of our knowledge, dates back to the pioneering work by J. L. Willems in 1974 \cite{willems1974partial}, the progress seems stagnant. Despite several mathematical concepts beyond Lyapunov stability have been proposed, e.g., partial stability \cite{vorotnikov2005partial}, set stabilization \cite{shiriaev2000stabilization} and contraction analysis \cite{lohmiller1998contraction}, they rarely find proper applications in power systems (studies handling phase rotational symmetry are among  few exceptions\cite{willems1974partial,jouini2020steady,colombino2019global}).
	A more favorable concept is \emph{synchronization}, which has been drawing increasing attention recently \cite{colombino2019global,9163825,zhu2018stability,motter2013spontaneous,dorfler2013synchronization,dorfler2012synchronization,paganini2019global}. 
	%It essentially redefines the ``nominal motion" in stability analysis from equilibrium to a synchronized state, which could make equilibrium-independent analytics possible. 
	It focuses on the synchronized state in stability analysis, instead of a given equilibrium. However, most existing works on this topic concern only synchrony among generators and are built on network-reduced ODEs models that assume constant voltages \cite{zhu2018stability,motter2013spontaneous,dorfler2013synchronization,dorfler2012synchronization}. Such models cannot capture the dynamical behavior of voltages, and fall short to capture heterogeneous devices in modern power grids. A more compatible and equilibrium-independent theory is still in need.
	
	A clear gap between theory and practice stands before us. On the one hand,  practical experience indicates that we can assess transient stability in an equilibrium-independent fashion but without knowing why. On the other hand, the equilibrium-dependent Lyapunov stability theory often encounters limitations in practice. Here, we aim to bridge this theory-practice gap by introducing the concept of \emph{augmented synchronization}, which means all frequencies synchronize and all voltages converge to steady states. 
	%	(mathematical definition will be given in Section \ref{sec:2}). 
	Inspired by the practice wisdom, we re-define the nominal motion of power systems as an augmented synchronous state instead of an equilibrium. And we re-interpret power system transient stability as convergence to augmented synchronization after a disturbance rather than to any specified post-disturbance equilibrium. Our interpretation conforms to the physical definition of power system transient stability and more importantly, will allow equilibrium-independent analytics.
	
	To this end, we aim to answer two questions in this paper. First, under what conditions does the convergence to augmented synchronization imply the convergence of all states? The answer will provide a theoretical justification for using augmented synchronization as a possible cornerstone in transient stability analysis, and also explain why the practice wisdom works. Second, given a system, from which initial points will the solution achieve augmented synchronization? Pursuing the answer is expected to stimulate a new equilibrium-independent approach to power system stability analysis without knowing the exact post-disturbance equilibrium \textit{a priori}. 
	Our main contributions are summarized as follows.
	\begin{itemize}
		\item We formulate the long-observed practice wisdom by introducing the concept of \emph{augmented synchronization}. Previous studies on power system synchronization\cite{9163825,zhu2018stability,motter2013spontaneous,dorfler2013synchronization,dorfler2012synchronization} usually adopt network-reduced ODE models and omit voltage dynamics. In contrast, our formulation is built on the structure-preserving model described by DAEs, accounting for both phase and voltage dynamics. Such formulation provides a framework to define and study synchronization among heterogeneous dynamical devices. Then, we propose sufficient conditions under which the convergence to augmented synchronization implies the convergence of all states to the equilibrium set, which we refer to as augmented synchronization detectability. This concept is weaker than the most widely-used ones such as  zero-state detectability \cite{100932} and output-to-state stability \cite{sontag1997output}. The major difference is that we do not expect to detect the states converging to any specific equilibrium. This result provides the first theoretical foundation for the long-observed practice wisdom that one may assess transient stability from only the observation of augmented synchronization.
		\item We establish theorems that estimate the region of attraction to augmented synchronization. That is, the region starting from which the solution of a power system will asymptotically converge to augmented synchronization. This result  extends the classical direct methods that estimate the region of attraction to a given post-disturbance equilibrium by the sublevel set of a Lyapunov function or an energy function \cite{chiang2011direct,hill1990stability,Chiang_Directstabilityanalysis_1995}. Here, the challenging difference is that we do not have prior knowledge of the post-disturbance equilibrium and hence our result concerns the convergence to a set rather than to any specific working point. Our result is also  different from the set stability theory \cite{fradkov2013nonlinear} because the set corresponding to augmented synchronization is often not an invariant set. This result opens opportunities for power system stability analysis by re-defining the ``nominal motion'' of a power system as augmented synchronization, rather than a single equilibrium as have been done for decades.
	\end{itemize}
	
	Combining the  two results above is expected to cultivate an \emph{equilibrium-independent} power stability analytic that requires no prior knowledge of the post-disturbance equilibrium. This could be a timely and more favorable approach, as power systems are evolving into more complex systems that operate in more volatile conditions, considering the burst of renewable generation integration. 
	
	The rest of the paper is organized as follows. Section II formulates the problem and introduces the concept of augmented synchronization. Section III addresses the first problem, giving conditions under which convergence to augmented synchronization implies convergence to the equilibrium set. Section IV addresses the second problem, providing methods to estimate the region of attraction to augmented synchronization without prior knowledge of the post-disturbance equilibrium. Section V illustrates the potential implications of our results using rudimentary examples. Section VI concludes the paper. 
	
	\emph{Notations}: $\rr$ ($\rr_{+}$) is the set of (positive) real numbers. Let $\text{col}(x_1,x_2)=(x_1^T,x_2^T)^T$ be a column vector in $\rr^{n+m}$ with $x_1\in\rr^n$ and $x_2\in\rr^m$. For a matrix $A\in\rr^{n\times n}$, $\mathrm{det}(A)$ denotes the determinant of $A$. For a symmetric  matrix $A\in\rr^{n\times n}$, $A\succ(\prec)0$ means $A$ is positive definite (resp. negative definite). $\mathds{1}_n\in\rr^n$ denotes the all-one vector.
	For a $v\in\rr^n$, $\|v\|$ denotes the 2-norm of $v$. And for a matrix $A\in\rr^{n\times m}$, $\|A\|$ denotes the induced 2-norm of $A$. When the context is clear, we may use 0 to denote a all-zero vector of a proper dimension.
	%Given another symmetric  matrix $B\in\rr^{n\times n}$, $A>(\geq)B$ means $A-B>(\geq)0$.
	\section{Problem Formulation}\label{sec:2}
	\subsection{Power System DAEs Model}
	We consider the structure-preserving power system \cite{bergen1981structure} described by the following DAEs
	\begin{subequations}\label{eq:system}
		\begin{align}
		\dot{x}_1&=f_1(x_1,x_2,z) \label{eq:system1a}\\
		\dot{x}_2&=f_2(x_1,x_2,z) \label{eq:system1b}\\
		0&=g(x_2,z), \label{eq:system1c}
		\end{align}
	\end{subequations} 
	with compatible initial conditions $(x_{1_0},x_{2_0},z_0)$, i.e., $0=g(x_{2_0},z_0)$. Here
	\begin{equation*}
	\begin{aligned}
	f_1&:\mathbb{R}^{n_1}\times\mathbb{R}^{n_2}\times\mathbb{R}^{m}\to\mathbb{R}^{n_1}\\
	f_2&:\mathbb{R}^{n_1}\times\mathbb{R}^{n_2}\times\mathbb{R}^{m}\to\mathbb{R}^{n_2}\\
	g&:\mathbb{R}^{n_2}\times\mathbb{R}^{m}\to\rr^{m}.
	\end{aligned}
	\end{equation*} Here $x_1\in\mathbb{R}^{n_1}$ denotes the state variables of a power system that do not appear in the algebraic equations \eqref{eq:system1c}, $x_2\in\mathbb{R}^{n_2}$ denotes the other state variables. For the simplicity of notation, in the case when all state variables appear in $g$, we write $x_1=\emptyset$ and all the following results in this paper still apply. Collectively, let $x:=\mathrm{col}(x_1,x_2)\in\rr^{n}$ denote the vector of all state variables where $n:=n_1+n_2$. Here $z=\mathrm{col}(\theta,V)\in\rr^m$ denotes the algebraic variables of a power system that are the voltage phases $\theta$ and voltage magnitudes $V$ of power network buses. Let $f:=\mathrm{col}(f_1,f_2)$. For simplicity, sometimes we may shorten $(x_1,x_2,z)$ as $(x,z)$ and may write $f(x,z)$, $f_1(x,z)$, and $f_2(x,z)$ representing $f(x_1,x_2,z)$, $f_1(x_1,x_2,z)$, and $f_2(x_1,x_2,z)$, respectively. 
	Let $$\mathbf{G}:=\{(x,z)\in\rr^{n}\times\rr^m|g(x_2,z)=0\},$$ denote the set of all points satisfying \eqref{eq:system1c}. We use $\left( x(t;x_0,z_0),z(t;x_0,z_0)\right)$ to denote the solutions of \eqref{eq:system} as a function of $t$ and initial conditions. When the initial conditions are irrelevant and clear, we may shorten the notation as $\left( x(t),z(t)\right)$.
	
	We assume throughout that functions $f_1$, $f_2$, and $g$ are all twice continuously differentiable w.r.t. $x_1$, $x_2$, and $z$, which are generally satisfied in power systems. And we assume in some open connected set $\mathcal{D}\subset\rr^{n}\times\rr^m$, the following assumption holds.
	\begin{assumption}\label{as:1}
		For any $(x,z)\in\overline{\mathcal{D}}$, $\mathrm{det}(\frac{\partial g}{\partial z})\neq0$, where $\overline{\mathcal{D}}$ denotes the closure of $\mathcal{D}$.
	\end{assumption}
	Let $$\mathcal{D}_G:=\mathcal{D}\cap\mathbf{G}.$$
	%  	Let $\mathbf{P}_z(\mathcal{D}_G)$ denote the projection of $\mathcal{D}_G$ on subspace $z$, that is, $$\mathbf{P}_z(\mathcal{D}_G):=\{z\in\mathbb{R}^m|\exists x\in\mathbb{R}^n, (x,z)\in\mathcal{D}_G\}.$$

	\begin{remark}
		Assumption \ref{as:1} together with the twice continuously differential property of $f$ and $g$ guarantee the existence and uniqueness of the solution to \eqref{eq:system} given any initial value $(x_0,z_0)\in D_G$ \cite{hill1990stability}. In fact, Assumption \ref{as:1} brings the most simple situation of a DAEs system, i.e., index-1, where the algebraic equation can be solved, at least locally, for $z$ as a function of $x$ \cite{kunkel2006differential}. Here, the non-singularity is required on the closure of $\mathcal{D}$ to simplify technical details such that ensure limit points in $\mathcal{D}$ are also non-singular. In power systems, Assumption \ref{as:1} is violated in the so-called impasse surface defined by $\mathrm{det}(\frac{\partial g}{\partial z})=0$. Solutions approaching the impasse surface are typically associated with the short-term voltage collapse in power systems \cite{hiskens1989energy,venkatasubramanian1995local}, which is beyond the scope of this paper.
	\end{remark}
	
	Define the function
	\begin{equation}\label{eq:h}
	h(x,z):=-\left(\frac{\partial g}{\partial z}(x_2,z)\right)^{-1}\frac{\partial g}{\partial x_2}(x_2,z)f_2(x,z),
	\end{equation}
	then $h$ is a continuously differentiable function on $\cal D$ by Assumption \ref{as:1} and $f_2,g\in C^2$. Consider the following ODEs system:
	\begin{subequations}\label{eq:systemode}
		\begin{align}
		\dot{x}&=f(x,z) \\
		\dot{z}&=h(x,z). \label{eq:zdot}
		\end{align}
	\end{subequations}
	For every $(x_0,z_0)\in\mathcal{D}$, the ODEs system has a unique solution defined on some interval $[0,t_+)\subset\mathbb{R}$. Clearly, $g(x_2,z)$ is constant along such solution and hence can be considered as a first integral of \eqref{eq:systemode}. It was shown in \cite{hill1990stability} that any solution of \eqref{eq:systemode} with $(x_0,z_0)\in\mathcal{D}_G$ is equivalent to the solution of \eqref{eq:system} in $\mathcal{D}_G$ with the same $(x_0,z_0)$. Note, however, the maximal $t_+$ for a solution of \eqref{eq:systemode} may be finite as the corresponding solution of \eqref{eq:system} may leave $\mathcal{D}_G$. While for solutions of \eqref{eq:system} that always stay in $\mathcal{D}_G$, it is equivalent to the corresponding solutions of \eqref{eq:systemode} with $t_+=\infty$. These arguments indicate that the DAEs system \eqref{eq:system} can be imbedded in the ODEs system \eqref{eq:systemode} with compatible initial conditions \cite{hill1990stability}. This allows us to apply results established for ODEs systems. 
	
	Further, we assume the DAEs system \eqref{eq:system} possesses at least one equilibrium in $\mathcal{D}_G$, which obviously is also an equilibrium of the ODEs system \eqref{eq:systemode}. Define the equilibrium set $$\mathcal{E}:=\big\{(x,z)\in\mathcal{D}_G|f(x,z)=0\big\}.$$
	\begin{assumption}
		$\mathcal{E}$ is not empty.
	\end{assumption}
	\subsection{Augmented Synchronization}
	Traditionally, power system transient stability is analyzed in the framework of Lyapunov stability. It concerns the asymptotic convergence of all states to an equilibrium. Given an equilibrium $(x^*,z^*)\in\mathcal{D}_G$ of \eqref{eq:system}, one is interested in from which initial point $(x_0,z_0)\in\mathcal{D}_G$, $(x(t),z(t))\to (x^*,z^*)$ as $t\to\infty$, assuming the existence of the solution.
	
	However, practitioners often focus on the convergence to synchronous frequencies and steady voltages. It concerns neither the convergence of other states nor to which equilibrium they converge. Here, we formulate this interesting and effective experience by introducing the concept of \textit{augmented synchronization}, defined as follows. 
	\begin{definition}
		A solution of \eqref{eq:system} is said to be an augmented synchronous solution if for all $t\geq0$, $(x(t),z(t))\in\mathcal{D}_G$ and
		\begin{equation*}
		\dot{z}(t)=\mathrm{col}(\dot{\theta}(t),\dot{V}(t))\equiv0.
		\end{equation*}	
	\end{definition}
	\begin{definition} \label{def:1}
		A solution of \eqref{eq:system} is said to be converging to augmented synchronization if for all $t\geq0$, $(x(t),z(t))\in\mathcal{D}_G$ and
		\begin{equation*}
		\dot{z}(t)=\mathrm{col}(\dot{\theta}(t),\dot{V}(t))\to0,\;\mathrm{as}\; t\to\infty.
		\end{equation*}	
	\end{definition}
	The prepositional adjective \textit{augmented} distinguishes our definition from the convention frequency synchronization. It emphasizes that not only synchronous frequencies but also steady voltages are required. Generally, frequency synchronization holds if there is a common frequency $\omega_s\in\rr$ such that $\dot{\theta}(t)=\omega_s\mathds{1}_n$. By working in a rotating framework, without loss of generality, we assume $\omega_s=0$. 
	\begin{remark}
		Compared with previous studies on synchronization that were built on ODEs, our formulation takes a different perspective. We define synchronization by the algebraic variables instead of states. Physically, that means we regard synchronization as a property of electrical sinusoidal voltages at all buses across the grid rather than a property among generators. Such formulation accounts for the voltage dynamics and provides a unified framework to incorporate heterogeneous devices such as power and current dynamics, which cannot be handled in previous studies \cite{9163825,zhu2018stability,motter2013spontaneous,dorfler2013synchronization,dorfler2012synchronization}, since there is no phase or voltage state in such dynamics. Our formulation seems more compatible and desirable for future power systems, considering generators are gradually giving place to inverter-interfaced heterogeneous devices. 
	\end{remark}
	
	\subsection{Problem Statement}
	With the above formulations, we now re-state the two questions that we aim to answer in this paper.
	
	Consider a solution $(x(t),z(t))$ of \eqref{eq:system} and the following two properties of the solution.
	
	\textit{Property 1:}  The solution stays in $\mathcal{D}_G$ for all $t\geq0$ and it holds that $z(t)$ is bounded and $\dot{z}(t)\to0$, as $t\to\infty$.
	
	\textit{Property 2:} The solution satisfies that 
	\begin{equation}\label{eq:converge2}
	(x(t),z(t))\to\mathcal{E},\;\text{as}\;\; t\to\infty,
	\end{equation}
	where the convergence to a set is defined in the sense of the distance to the set converging to zero.
	
	First, we are interested in under what conditions the practice wisdom, which assesses transient stability by only observation of $z(t)$, i.e., \textit{Property 1}, would imply the convergence to the equilibrium set, i.e. \textit{Property 2}. Obviously, the former is a necessary condition for the latter, but generally, it is not sufficient. This problem is stated as follows.
	\begin{problem}\label{problem:1} 
		Consider a solution $(x(t),z(t))$ of \eqref{eq:system}. Under what conditions does \textit{Property 1} of the solution imply \textit{Property 2}?
	\end{problem}
	
	Note that \eqref{eq:converge2} implies the convergence of $\dot{x}(t)$ to zero, i.e.,
	\begin{equation}\label{eq:converge1}
	\dot{x}(t)=f(x(t),z(t))\to0,\;\text{as}\;\; t\to\infty,
	\end{equation}
	which has a natural physical meaning from the engineering point of view, i.e., it requires the system reach steady state. But generally \eqref{eq:converge1} does not imply \eqref{eq:converge2} if the solution is not bounded. Note also since the solution stays in $\mathcal{D}_G$, it is equivalent to study Problem \ref{problem:1} in the ODEs system \eqref{eq:systemode}. 
	%	Consider, for example, a two-dimensional system
	%	\begin{equation*}
	%	\dot{x}=-x+x^2y,\qquad \dot{y}=y-xy^2 
	%	\end{equation*}
	%	Given the initial value $x(0)=x_0$, 
	%   $y(0)=y_0$, the solution is
	%	\[x(t)=x_0e^{-t(1-x_0y_0)},\quad   
	%    y(t)=y_0e^{t(1-x_0y_0)}\]
	%	For $x_0y_0<1$, one will observe $x(t)\to0$ as $t\to\infty$, while $y(t)$ approaches infinity. For $x_0y_0>1$ the opposite phenomenon shows up. This illustrates the convergence of part of states does not necessarily imply the other.
	%\begin{remark}
	%	We remark that $\lim_{t\to\infty}\dot{z}(t)=0$ and $z(t)$ having a finite limit at $t\to\infty$ are both required. For a general function of time $f(t)$, $\dot{f}(t)\to0$ as $t\to\infty$ does not imply $f$ has a finite limit, e.g., $f(t)=\ln(t)$ and $f(t)=\sin(\ln(t))$. On the other hand, $f(t)$ has a finite limit at $t\to\infty$ does not imply $\dot{f}(t)\to0$ neither, e.g., $f(t)=\frac{\sin(t^2)}{t}$.
	%\end{remark}
	
	Second, we are interested in under what conditions the power system \eqref{eq:system} can achieve augmented synchronization. Specifically, we consider the follows problem.
	\begin{problem}\label{problem:2}
		From which initial point $(x_0,z_0)\in\mathcal{D}_G$ does the solution $(x(t),z(t))$ of \eqref{eq:system} satisfy \textit{Property 1}?
	\end{problem}  
	In fact, the classical equilibrium-based Lyapunov stability analysis provides an answer to this question with prior knowledge of the post-disturbance equilibrium. Here, we aim to tackle this question in an equilibrium-independent way.
	
	The first question will be addressed in Section \ref{sec:explain} and the second  in Section \ref{sec:extend}.  %One appealing implication of 

	\section{Why Practice Wisdom works}\label{sec:explain}
	This section provides our results of Problem \ref{problem:1}.
	%	Suppose we observe $\dot{z}(t)\to0$ as $t\to\infty$, and we assume $z(t)\to z^*$ as $t\to\infty$ for some possibly unknown $z^*$. Practitioners would claim the power system \eqref{eq:system} is transiently stable. But why the convergence of $z(t)$ can imply the convergence of $x(t)$? 
	We find that what underlies the practice wisdom is a widely satisfied property of power systems, which we refer to as \textit{augmented synchronization detectability}. In the following part, we present a rigorous definition of this concept and several checkable criteria with illustrative examples.
	\subsection{Augmented Synchronization Detectability}
	Throughout this section, we will assume that we have observed a solution of \eqref{eq:system} that satisfies \textit{Property 1}. It can be viewed as certain detectability-type property of the solution if \textit{Property 1} implies \textit{Property 2}. We introduce the concept of augmented synchronization detectability to capture this property, which is defined as follows.
	\begin{definition}
		We say a solution $(x(t),z(t))$ of \eqref{eq:system} is augmented synchronization detectable (AS-detectable) if \textit{Property 1} of the solution implies \textit{Property 2}.
	\end{definition}
	
	\begin{remark}
		Augmented synchronization detectability is closely related to the well-known concepts of zero-state detectability \cite{100932} and  the output-to-state stability \cite{sontag1997output}. All three concepts are relevant to the property that one can infer the behavior of all states solely based on the observation of part of states or a certain function of states, e.g., the outputs. However, compared with zero-state detectability and output-to-state stability, AS-detectability is much weaker and requires neither $z(t)$ nor $x(t)$ converging to some specified points. Instead, it only focuses on the convergence of their time derivatives $\dot{z}(t)$ and $\dot{x}(t)$ to zero, which enables equilibrium-independent analytics.
	\end{remark}
	The remainder of this section is devoted to finding checkable conditions under which AS-detectability holds, which provides the practice wisdom with a solid theoretical foundation.
	\subsection{Conditions for Non-Degenerate Solutions}
	We begin with a simple situation where the solution is non-degenerate, defined as follows.
	\begin{definition}\label{def:nondegenerate}
		Suppose $(x(t),z(t))$ is a solution of \eqref{eq:system} that stays in $\mathcal{D}_G$ for all $t\geq0$. We say the solution is non-degenerate if $\frac{\partial g}{\partial x_2}$ has constant a full column rank on the solution, i.e., $$\text{rank}\left(\frac{\partial g}{\partial x_2}(x_2(t),z(t))\right)=n_2,\;\; \forall t\geq0$$ and the matrix $(\frac{\partial g}{\partial x_2})^{\dagger}\frac{\partial g}{\partial z}$ is bounded, i.e., $\exists M>0$ s.t.
		$$\left\|\frac{\partial g}{\partial x_2}^{\dagger}(x_2(t),z(t))\frac{\partial g}{\partial z}(x_2(t),z(t))\right\|\leq M,\;\; \forall t\geq0$$
		where $\frac{\partial g}{\partial x_2}^{\dagger}:=\left(\frac{\partial g}{\partial x_2}^T\frac{\partial g}{\partial x_2}\right)^{-1}\frac{\partial g}{\partial x_2}^T$ denotes the left inverse. Otherwise, we say a solution is degenerate.
	\end{definition}	
	For non-degenerate solutions a direct connection between $\dot{z}$ and $\dot{x}_2$ holds, as stated in the following Lemma.
	\begin{lemma}\label{le:z-x2}
		Consider a non-degenerate solution of \eqref{eq:system} that satisfies \textit{Property 1}. It holds that $\lim_{t\to\infty}\dot{x}_2(t)=0$ and $x_2(t)$ is bounded.
	\end{lemma}
	\begin{proof}
		Since the solution is non-degenerate, it follows from \eqref{eq:h} that
		\begin{equation}\label{eq:x2dot}
		\dot{x}_2(t)=-\left(\frac{\partial g}{\partial x_2}\right)^{\dagger}\frac{\partial g}{\partial z}\dot{z}(t).
		\end{equation}
		Since $\|(\frac{\partial g}{\partial x_2})^{\dagger}\frac{\partial g}{\partial z}\|<\infty$, $\dot{z}(t)\to0$ implies $\dot{x}_2(t)\to0$. And for any $t\geq 0$,
		\begin{equation*}
		\left\| \int_{0}^{t}\dot{x}_2(\tau)\mathrm{d}\tau\right\|\leq\sup\left\|(\frac{\partial g}{\partial x_2})^{\dagger}\frac{\partial g}{\partial z}\right\|\|z(t)-z(0)\|.
		\end{equation*}
		Hence, it follows from the boundedness of $z(t)$ that $x_2(t)$ is bounded.
	\end{proof}
	If all state variables appear in the algebraic equations, i.e., $x_1=\emptyset$, then non-degenerate solutions are AS-detectable.
	\begin{theorem}\label{th:1}
		Suppose $x_1=\emptyset$. Then every non-degenerate solution of \eqref{eq:system} is AS-detectable.
	\end{theorem}
	\begin{proof}
		Provide the solution is non-degenerate and \textit{Property 1}. Lemma \ref{le:z-x2} indicates the solution $(x_2(t),z(t))$ is bounded. It follows from \cite[Lemma 4.1]{khalil2002nonlinear} that the $\omega$-limit set $L^{+}$ of the solution $(x_2(t),z(t))$ is a nonempty, compact, and invariant set. And $(x_2(t),z(t))$ approaches $L^+$ as $t\to\infty$. Let $(x^*_2,z^*)$ be an $\omega$-limit point of the solution. There is an increasing sequence of time $\{t_n\}$ such that $\lim_{n\to\infty}t_n=\infty$ and  $\lim_{n\to\infty}(x_2(t_n),z(t_n))=(x_2^*,z^*)$. By continuity of $f_2$,
		\[\lim_{n\to\infty}\dot{x}_2(t_n)=\lim_{n\to\infty}f_2(x_2(t_n),z(t_n))=f_2(x^*_2,z^*)=0.\]
		This indicates the $\omega$-limit set only consists of equilibria, which implies \textit{Property 2}.
	\end{proof}
	When $x_1\neq\emptyset$, AS-detectability can be verified given two additional conditions of $x_1(t)$.
	
	\begin{theorem}\label{th:auto2}
		Suppose $x_1\neq\emptyset$. A non-degenerate solution of \eqref{eq:system} is AS-detectable if:
		\begin{enumerate}
			\item[1)] \textit{Property 1} implies the boundedness of $x_1(t)$;
			\item[2)] No solution $(x(t),z(t))$ of \eqref{eq:system} can stay identically in $\mathcal{M}$ other than the solutions such that $f_1(x(t),z(t))\equiv0$, where $\mathcal{M}$ is defined as
			\begin{equation}\label{eq:M2}
			\mathcal{M}:=\left\lbrace (x,z)\in\mathcal{D}_G|f_2(x,z)=0\right\rbrace.
			\end{equation}
		\end{enumerate}
	\end{theorem}
	\begin{proof}
		Given \textit{Property 1}, it follows from Lemma \ref{le:z-x2} and condition 1) that the solution $(x(t),z(t))$ is bounded. Hence, $(x(t),z(t))$ approaches its $\omega$-limit set $L^+$ as $t\to\infty$.
		
		%		Now consider the exact $z^*$ given in \textit{Property 1}. 
		For any $(x^*,z^*)\in L^+$, it follows from the same argument in the previous proof that $f_2(x^*,z^*)=0$.
		Let $\Upsilon$ be the largest invariant set of $\mathcal{M}$. Thus, \[L^+\subset\Upsilon\subset \mathcal{M}.\]
		
		Condition 2) further guarantees for any $(x^*,z^*)\in\Upsilon$, $f(x^*,z^*)=0$. 
		Since the solution approaches $L^+$ as $t\to\infty$, it approaches $\Upsilon$ as well. Hence, the solution approaches $\mathcal{E}$ as $t\to\infty$, which implies \textit{Property 2}.
	\end{proof}
	%	Before we move on to the more complicated degenerate situations, 
	We next show how to use the previous conditions to verify AS-detectability in two examples that are prevalent in modern power systems: an inverter-interfaced power source and a synchronous generator.
	
	\begin{example}\label{ex:ps}
		Consider an inverter-interfaced power source connected to the infinite bus via a transmission line, as shown in Fig.\ref{fig:SMIB_inverter}. We simplify the dynamics of the inverter, and assume it regulates the output power according to the deviation of the terminal phase and voltage via first-order dynamics.
		\begin{figure}[h]
			\centering
			\setlength{\abovecaptionskip}{0.2cm}	
			\setlength{\belowcaptionskip}{-0.2cm}
			\includegraphics[width=.6\hsize]{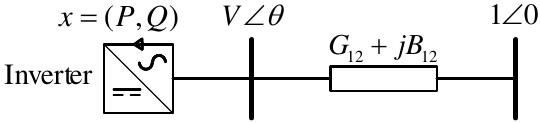}
			\caption{An inverter-interfaced power source connecting to the infinite bus.}
			\label{fig:SMIB_inverter}
		\end{figure}
		
		The system dynamics read
		\begin{equation}\label{eq:ps}
		\left\lbrace\begin{aligned}
		&\tau_1\dot{P}=-P+P^{\mathrm{ref}}-d_1(\theta-\theta^{\mathrm{ref}})\\
		&\tau_2\dot{Q}=-Q+Q^{\mathrm{ref}}-d_2(V-V^{\mathrm{ref}}),
		\end{aligned} \right.
		\end{equation}
		where $P$ and $Q$ are the terminal output active and reactive power, respectively. $V\angle\theta$ is the terminal complex voltage. $P^{\mathrm{ref}}$, $Q^{\mathrm{ref}}$, $\theta^{\mathrm{ref}}$, and $V^{\mathrm{ref}}$ are pre-specified constant reference values. $\tau_1>0$ and $\tau_2>0$ are time constants. $d_1>0$ and $d_2>0$ are droop coefficients. The algebraic equations read
		\begin{equation}\label{eq:ps-g}
		\left\lbrace \begin{aligned}
		P-G_{11}V^2-G_{12}V\cos\theta-B_{12}V\sin\theta=0\\
		Q+B_{11}V^2-G_{12}V\sin\theta+B_{12}V\cos\theta=0.
		\end{aligned}\right. 
		\end{equation}
		
		\begin{proposition}\label{pro:dg}
			Every solution of the power system in Example \ref{ex:ps} satisfying \textit{Property 1} is non-degenerate. And any solution of the system is AS-detectable.
		\end{proposition}
		\begin{proof}
			In this case, $x_1=\emptyset$. It follows from \eqref{eq:ps-g} that
			\[\frac{\partial g}{\partial x_2}=\begin{bmatrix}
			1&0\\0&1
			\end{bmatrix}.\]
			And one can verify that $\frac{\partial g}{\partial z}$ is bounded for any bounded $z$. Hence, any solution that satisfies \textit{Property 1} is non-degenerate. Then, by Theorem \ref{th:1}, any solution is AS-detectable.
		\end{proof}
	\end{example}
	\begin{example}\label{ex:sg}
		Consider a synchronous generator (SG) connected to an infinite bus via a transmission line, as shown in Fig.\ref{fig:SMIB}.
		\begin{figure}[h]
			\centering
			\setlength{\abovecaptionskip}{0.2cm}	
			\setlength{\belowcaptionskip}{-0.2cm}
			\includegraphics[width=.6\hsize]{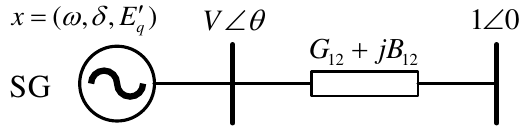}
			\caption{A synchronous generator connecting to the infinite bus.}
			\label{fig:SMIB}
		\end{figure}
		
		The system dynamics read \cite{1085625}
		\begin{equation}\label{eq:sg}
		\left\lbrace\begin{aligned}
		&\dot{\delta}=\omega\\
		&M\dot{\omega}=-D\omega-P^e+P^g\\
		&T'_{d0}\dot{E}'_{q}=-\frac{x_{d}}{x'_{d}}E'_{q}+\frac{(x_{d}-x'_{d})V\cos(\delta-\theta)}{x'_{d}}+E^f,
		\end{aligned} \right.
		\end{equation}
		where $E_{q}'\angle\delta$ is the $q$-axis transient internal complex voltage and $V\angle\theta$ is the terminal complex voltage. $\omega$ is the frequency derivation. $M$ is the moment of inertia. $D$ is the damping coefficient. $T_{d0}'$ is the $d$-axis open-circuit transient time constant. $x_{d}$, $x_{d}'$, and $x_q$ are the $d$-axis synchronous reactance, $d$-axis transient reactance, and $q$-axis synchronous reactance, respectively. For a realistic SG, $x_{d}>x_{d}'$.
		%		For salient machines $x_d>x_q$ and for non-salient machines $x_d=x_q$. 
		$P^g$ and $E^f$ are constant parameters meaning the power generation and the excitation voltage, respectively. The terminal output active power $P^e$ and reactive power $Q^e$ are given by
		\begin{equation}\label{eq:sg-Pe}
		P^e=\frac{V^2\sin[2(\delta-\theta)](x'_{d}-x_{q})}{2x_{q}x'_{d}}+\frac{E'_{q}V\sin(\delta-\theta)}{x'_{d}},
		\end{equation}
		and
		\begin{equation}\label{eq:sg-Qe}
		\begin{aligned}
		Q^e=&\frac{V^2\cos[2(\delta-\theta)](x'_{d}-x_{q})}{2x_{q}x'_{d}}+\frac{E'_{q}V\cos(\delta-\theta)}{x'_{d}}\\&-\frac{x'_{d}+x_{q}}{2x_{q}x'_{d}}V^2.
		\end{aligned}
		\end{equation}
		Power balance constraints at the terminal bus give the algebraic equations:
		\begin{equation}\label{eq:sg-g}
		\left\lbrace \begin{aligned}
		P^e-G_{11}V^2-G_{12}V\cos\theta-B_{12}V\sin\theta=0\\
		Q^e+B_{11}V^2-G_{12}V\sin\theta+B_{12}V\cos\theta=0,
		\end{aligned}\right.
		\end{equation}
		where $G_{ij}$ and $B_{ij}$ are elements of the admittance matrix corresponding to the transmission line.
		
		\begin{proposition}\label{pro:sg}
			Every non-degenerate solution of the power system in Example \ref{ex:sg} is AS-detectable.
		\end{proposition}
		\begin{proof}
			In this case, $x_1=\omega$, $x_2=(\delta,E'_q)^T$, and $z=(\theta,V)^T$.
			For condition 1) in Theorem \ref{th:auto2}, it suffices to prove $\omega(t)$ is bounded when $x_2(t)$ is bounded. Since $\delta(t)=\int_{0}^{t}\omega(\tau)\mathrm{d}\tau+\delta(0)$, it follows from continuity of $\omega(t)$ that $\omega(t)$ must be bounded if $\delta(t)$ is bounded.			
			For condition 2) in Theorem \ref{th:auto2}, it suffices to prove $f_1=\dot{\omega}(t)=0$ on the largest invariant set of $\mathcal{M}$. Since $\dot{\delta}(t)=\omega(t)=0$ on $\mathcal{M}$, it holds that $\omega(t)\equiv0$ on the invariant set of $\mathcal{M}$. Hence, $\dot{\omega}(t)=0$ on the invariant set of $\mathcal{M}$.
			Hence, it follows from Theorem \ref{th:auto2} that every non-degenerate solution of this system is AS-detectable.
		\end{proof}
	\end{example}

	\subsection{Compositional Conditions for Modular Structured Power Systems}\label{subsec:modular}
	Although we have only illustrated the usage of previous conditions on two single machine systems, it should be noted that those conditions apply to general large-scale power systems. However, it might be a difficult task to determine the boundedness of $x_1(t)$ and the invariant set in $\cal M$ for a general large system. In this subsection, we will propose a compositional approach to AS-detectability for modular structured power systems, which requires only conditions of each low-dimensional subsystem.
	
	%, and each bus is connected with a dynamic subsystem ($x\neq\emptyset$) such as a generator, or a static subsystem ($x=\emptyset$) such as a constant PQ load (if $P=Q=0$, it represents an interior connecting bus)
	
	Suppose a large-scale power system consisting of finite $n$ buses. Let $i\in\{1,\ldots,n\}$ be the bus index. We say the power system is modular structured if there is a decomposition of variables and functions, denoted as $x_{1,i}$, $x_{2,i}$, $z_i$, $f_{1,i}$, and $f_{2,i}$ for $i=1,\ldots,n$, such that: $x_1=\text{col}(x_{1,1},\ldots,x_{1,n})$, $x_2=\text{col}(x_{2,1},\ldots,x_{2,n})$, $z=\text{col}(z_1,\ldots,z_{n})$, $f_1=\text{col}(f_{1,1},\ldots,f_{1,n})$, and $f_2=\text{col}(f_{2,1},\ldots,f_{2,n})$; and
	for each $i$, $f_{1,i}$ and $f_{2,i}$ are only functions of local variables, i.e., $f_{1,i}(x_{1,i},x_{2,i},z_i)$ and  $f_{2,i}(x_{1,i},x_{2,i},z_i)$.
	
	Hence, for a modular structure power system, \eqref{eq:system1a} and \eqref{eq:system1b} can be equivalently written as the composition of subsystems: for $i=1,\ldots,n$
	\begin{equation}\label{eq:subsystem}
	\begin{aligned}
	\dot{x}_{1,i}&=f_{1,i}(x_{1,i},x_{2,i},z_i)\\
	\dot{x}_{2,i}&=f_{2,i}(x_{1,i},x_{2,i},z_i), 
	\end{aligned}
	\end{equation}
	Note, however, all subsystems are still coupled via the algebraic equation $0=g(x_2,z)$. With a little abuse of notation, we allow $x_{1,i}=\emptyset$ and $f_{1,i}=\emptyset$ for some $i$ to adapt to some subsystems, state variables of which are all involved in the algebraic equations. 
	
	The following theorem states that for a modular structure power system, conditions for AS-detectability of the entire system \eqref{eq:system} can be decomposed into conditions of each subsystem \eqref{eq:subsystem}.
	
	\begin{theorem}\label{pro:modular}
		Suppose the power system \eqref{eq:system} is modular structured. Every non-degenerate solution of \eqref{eq:system} is AS-detectable if for each $i$ with $x_{1,i}\neq\emptyset$ the subsystem \eqref{eq:subsystem} regarding $z_i(t)$ as input satisfies:
		\begin{enumerate}
			\item[1)] The boundedness of $z_i(t)$ and $x_{2,i}(t)$ implies the boundedness of $x_{1,i}(t)$;
			\item[2)] For any bounded input $z_i(t)$, no solution $x_i(t)$ of \eqref{eq:subsystem} exists, other than the solutions that $f_{1,i}(x_i(t),z_i(t))\equiv0$, such that $(x_i(t),z_i(t))$ stays identically in $\mathcal{M}_i$  where $\mathcal{M}_i$ is defined as
			\[\mathcal{M}_i:=\{(x_i,z_i)\in\rr^{n_i}\times\rr^2|f_{2,i}(x_{i},z_i)=0\}.\]
		\end{enumerate}
	\end{theorem}
	\begin{proof}
		Consider a non-degenerate solution $(x(t),z(t))$ of \eqref{eq:system} that satisfies \textit{Property 1}. It follows from Lemma \ref{le:z-x2} that $x_2(t)$ is bounded and $f_{2}(t)\to0$ as $t\to\infty$. Hence, for each $i$, $z_i(t)$ and $x_{2,i}(t)$ are bounded and $f_{2,i}(t)\to0$. By 1), $x_{1,i}(t)$ is bounded and hence $x_{1}(t)$ is bounded. It follows from the modular structure that $x_i(t)$ is also a solution to \eqref{eq:subsystem} under bounded input $z_i(t)$. Hence, by a similar argument, 2) implies that the $\omega$-limit set of $(x_i(t),z_i(t))$ must be the equilibrium set, which implies \textit{Property 2}.
	\end{proof}
	Theorem \ref{pro:modular} provides a compositional approach to verify AS-detectability for large-scale modular structure power systems. Since conditions 1) and 2) are local and depend on subsystem $i$ only, it makes no difference to check them in an interconnected system or in the single-machine-infinite-bus setting.
	
	The modular structure is ubiquitous in power systems, which are widely satisfied in the study of power system stability (see for example \cite{1085625,schiffer2014conditions,8890862}). Physically, this means each subsystem is controlled independently and interacts with others only by power flows. When inter-area control is employed, one may still treat an independent area as one modular block. An example of using this modular property to verify AS-detectability of the IEEE 9-bus system will be presented in Section \ref{subsec:9bus}.
	\begin{remark}
		By similar arguments as in Example \ref{ex:ps}-\ref{ex:sg}, one may verify that other power system dynamical devices satisfy conditions in Theorem \ref{pro:modular}. Combined with the ubiquitous modular structure of power systems, we conjecture that AS-detectability, at least for non-degenerate solutions, is a common property of power systems. 
	\end{remark}
	\begin{remark}
		Technically, a solution of \eqref{eq:system} that converges to an unstable equilibrium is indeed a solution that satisfies both \textit{Property 1} and \textit{Property 2}. Nevertheless, the possibility of such a solution taking place in practice is zero, since the stable manifold of an unstable equilibrium has zero measure in the state space \cite{357,zaborszky1988phase}. And hence an infinitesimal perturbation can cause the solution to diverge from the unstable equilibrium. This idea has led to the theory of quasi-stability region, which practically regards the stable manifold of unstable equilibrium in the interior of the stability region, as part of the quasi-stability region \cite{chiang1996quasi,chiang1996quasi2}. This fact, together with AS-detectability, provides a theoretical explanation for the long-observed practice wisdom that one may assess transient stability in practice solely based on the observation of augmented synchronization.
	\end{remark}
	\subsection{Physical Interpretation of Degeneration}
	This subsection will provide an interpretation of degeneration in the context of power systems, which helps to build insightful intuitions of AS-detectability.
	
	It is reasonable to assume that $\frac{\partial g}{\partial z}$ is bounded when the solution is bounded, in most power system models. Therefore, degeneration as defined in Definition \ref{def:nondegenerate} often takes place when $\frac{\partial g}{\partial x_2}$ has a deficient column rank. This indicates that an infinitesimal change of $x_2$, given by the null space of $\frac{\partial g}{\partial x_2}$, will not change the value of the function $g$. Physically, since $g$  represents the power balance at buses, degeneration implies the terminal output power of some devices is locally irrelevant to their internal states $x_2$. 
	
	To further illustrate, consider for example the single machine infinite bus power system as described in Example \ref{ex:sg}.
	It follows from \eqref{eq:sg-Pe}-\eqref{eq:sg-g} that
	\begin{equation*}
	\frac{\partial g}{\partial x_2}=\begin{bmatrix}
	\frac{V^2(x'_d-x_q)\cos[2(\delta-\theta)]}{x'_dx_q}+\frac{E'_{q}V\cos(\delta-\theta)}{x'_{d}}&\frac{V\sin(\delta-\theta)}{x'_{d}}\\
	-\frac{V^2(x'_d-x_q)\sin[2(\delta-\theta)]}{x'_dx_q}-\frac{E'_{q}V\sin(\delta-\theta)}{x'_{d}}&\frac{V\cos(\delta-\theta)}{x'_{d}}
	\end{bmatrix}.
	\end{equation*}
	It yields
	\begin{equation*}
	\mathrm{det}\left(\frac{\partial g}{\partial x_2}\right)=\frac{V^2}{x'^2_d}\left(\frac{x'_d-x_q}{x_q}V\cos(\delta-\theta)+E'_q\right).
	\end{equation*}
	Hence if $V\neq0$ and $\frac{x'_d-x_q}{x_q}V\cos(\delta-\theta)+E'_q\neq0$ on the solution, the solution is non-degenerate if it is bounded.
	
	Two possibilities of degeneration appear. First, $V=0$ takes place at some point on the closure of the solution. Since $V=0$ yields $\mathrm{det}(\frac{\partial g}{\partial z})=0$ as well, Assumption \ref{as:1} does not hold for this solution neither. Physically, this would only happen when the generator's terminal undergoes a purely metallic short-circuit fault. It is natural that for solutions satisfying $V(t)=0$ or $V(t)\to0$, no conclusion of $x(t)$ can be made, even though we observe the convergence of $z(t)$.
	
	Now we assume $V>0$ and turn to the second possibility of degeneration, i.e.,
	\begin{equation}\label{eq:remark-singular}
	\frac{x'_d-x_q}{x_q}V\cos(\delta-\theta)+E'_q=0.
	\end{equation}
	Substituting \eqref{eq:remark-singular} into \eqref{eq:sg-Pe} and \eqref{eq:sg-Qe}, we obtain
	$P^e=0$ and $Q^e=-\frac{1}{x_q}V^2$.
	This indicates that the terminal output power of the generator becomes independent of the system state $E'_q\angle\delta$, in which case the detectability from output to state is lost.
	%	Surprisingly, however, solutions having this type of degeneration are still AS-detectable, as will be shown in the next subsection.
	
	Similar results can be made for large-scale modular structured power systems that degeneration takes place when there exists one device whose output power becomes totally irrelevant to its internal states. Such scenarios should be rare in practical power systems.
	\section{When augmented synchronization}\label{sec:extend}
	In this section, we turn to the second question: under what conditions a power system can reach augmented synchronization? To this end, we consider a more general problem: given an output function $\eta(x,z)$, starting from what initial value $(x_0,z_0)\in\mathcal{D}_G$ the solution of \eqref{eq:system} would satisfy $\eta(x(t),z(t))\to0$ as $t\to\infty$. Clearly, solving the problem with $\eta=h$ would answer our second question. However, the choice of $\eta$ is not restricted to $h$, as long as $\eta\to0$ implies $h\to0$, which provides additional flexibility. For example, it is sometimes more convenient to set $\eta=f_2$ in power systems as will be shown in Section \ref{sec:case}.
	
	Based on this idea, we introduce the region of attraction to $\eta$ ($\eta$-RoA) defined as follows.
	\begin{definition}\label{def:roa}
		Consider the power system \eqref{eq:system}. Given a function $\eta:\mathcal{D}\to\rr^{n_\eta}$ for some dimension $n_\eta$, the $\eta$-RoA of \eqref{eq:system} is defined as
		\[R_\eta:=\left\{(x,z)\in\mathcal{D}_G\left|
		\begin{aligned}
		\lim_{t\to\infty}\eta((x(t),z(t)))=0,\\
		(x(0),z(0))=(x,z)
		\end{aligned}\right. \right\},\]
		where $(x(t),z(t))$ is the solution of \eqref{eq:system} with initial value $(x(0),z(0))$.
		We say a set $\Omega$ to be an estimation of the $\eta$-RoA if $\Omega\subset R_\eta$.
	\end{definition}
	%	We remark that only open and positively invariant set as $\eta$-RoA is of interest. This rules out trajectories converging to unstable saddle-type equilibria on stability boundaries, which may satisfy $\eta(x,z)\to0$ but are not open \cite{willems1974partial}.
	\begin{remark}
		Note that the above definition does not imply the set $\{(x,z)\in\mathcal{D}_G|\eta(x,z)=0\}$ to be an attracting or stable set as in the literature, which usually requires positive invariance of the set \cite[Chapter 4]{fradkov2013nonlinear}. Here, we do not require positive invariance nor attractivity in its neighborhood of the set. These relaxations are critical for power system applications, in which $\{(x,z)\in\mathcal{D}_G|h(x,z)=0\}$ is often not positively invariant.
	\end{remark}
	
	%Throughout this section, we assume that in power systems when we observe $\dot{z}(t)$ approaches zero, $z(t)$ would also approach some finite value.
	%\begin{assumption}\label{as:2}
	%	For the power system \eqref{eq:system}, if $\dot{z}(t)\to0$ as $t\to\infty$ then $z(t)$ has a finite limit as $t\to\infty$.
	%\end{assumption}
	%This assumption means whenever the changes of phases and voltages of all buses tend to zero, phases and voltages themselves will be bounded and tend to some finite values as well. This excludes some rare but strange behaviors, e.g., even if $\dot{z}(t)\to0$, $z(t)$ still becomes unbounded or $z(t)$ sustainedly oscillates. We conjecture this assumption is generally valid in real power systems as such behaviors have not been reported in practice so far. 
	
	Now suppose we have obtained an estimation of the $\eta$-RoA. The next theorem states that under some additional conditions solutions starting in $\eta$-RoA satisfy \textit{Property 1}. 
	\begin{theorem}\label{th:stable}
		Consider the power system \eqref{eq:system} satisfying Assumption \ref{as:1} and a function $\eta$ so that $\eta\to0$ implies $h\to0$ on $\mathcal{D}_G$. Let a set $\Omega\subset\mathcal{D}_G$ such that
		\begin{enumerate}
			\item $\Omega$ is positively invariant;
			\item $z$ is bounded on $\Omega$;
			\item $\Omega$ is an estimation of the $\eta$-RoA of \eqref{eq:system}.
		\end{enumerate}
		Then, every solution starting from $\Omega$ satisfies \textit{Property 1}. Further, if every solution in $\Omega$ is AS-detectable, then, every solution converges to the equilibrium set $\mathcal{E}$. %Further, if $\Omega$ contains no Lyapunov unstable equilibria, then all solutions in $\Omega$ converge to the Lyapunov stable equilibrium set.
	\end{theorem}
	\begin{proof}
		Since $\Omega$ is positively invariant, every solution starting from $\Omega$ stays in $\Omega\subset\mathcal{D}_G$ for all $t\geq0$, and hence by 2) $z(t)$ is bounded. By 3) $\Omega$ is an estimation of the $h$-RoA, and hence $\dot{z}\to0$ as $t\to\infty$. Therefore, such solution satisfies \textit{Property 1}. By AS-detectability, such solution converges to $\mathcal{E}$.
	\end{proof}
	This theorem enables an equilibrium-independent method to verify power system transient stability. If $\Omega$ contains no Lyapunov unstable equilibria, then all solutions in $\Omega$ converge to the Lyapunov stable equilibrium set. One can construct $\Omega$ and justifying AS-detectability without knowing the exact post-disturbance equilibrium. Then, one can assess transient stability by checking whether the post-disturbance initial value locates inside $\Omega$ or not.
	
	The rest of this section will present our results on how to estimate the $\eta$-RoA.
	Our results are built on Lyapunov-like $\cal V$-functions, which measure the distance between the current state to the desired situation, i.e., $\eta(x,z)=0$. Such $\mathcal{V}$-functions are natural extensions of Lyapunov functions with only a different perspective that treats $\eta(x,z)=0$, instead of an equilibrium, as the nominal motion.
	
	To streamline the presentation, we first introduce the following notations. For a $C^1$ function $\mathcal{V}:\rr^n\times\rr^m\to\rr$, let $\mathcal{V}_l^{-1}$ denote the intersection of its $l$-sublevel set and $\mathbf{G}$, defined as
	\[\mathcal{V}_l^{-1}:=\left\lbrace(x,z)\in\rr^n\times\rr^m|\mathcal{V}(x,z)\leq l\right\rbrace\cap\mathbf{G}.\]
	If the sublevel set contains several disjoint components, $\mathcal{V}_l^{-1}$ refers to one connected component.
	A continuous function $\alpha:[0,a)\to[0,\infty)$ is said to be a $\cal K$ function if it is strictly increasing and $\alpha(0)=0$ \cite{khalil2002nonlinear}.
	%\begin{definition}[Hassan]
	%	A continuous function $\alpha:[0,a)\to[0,\infty)$ is said to be a $\cal K$ function if it is strictly increasing and $\alpha(0)=0$.
	%\end{definition}
	
	We establish three theorems to provide estimations of $\eta$-RoA via three types of $\mathcal{V}$-functions, which differ in specific requirements of $\mathcal{V}$ and $\dot{\mathcal{V}}$.
	%	The remainder of this section will report our findings one by one and will discuss their suitable application scenarios.
	\subsection{Type I }
	\begin{theorem}\label{th:s1}
		If there exist a scalar $C^1$ function $\mathcal{V}:\rr^n\times\rr^m\to\rr$, a vector function $\xi:\rr^n\times\rr^m\to\rr^{n_\xi}$ for some dimension $n_\xi$, and $\cal K$ functions $\alpha$, $\beta$, and $\gamma$ such that for every $(x,z)\in\mathcal{D}_G$:
		\begin{enumerate}
			\item $\alpha(\|\eta(x,z)\|)\leq \mathcal{V}(x,z)\leq\beta(\|\xi(x,z)\|)$,
			\item $\dot{\mathcal{V}}(x,z)\leq-\gamma(\|\xi(x,z)\|)$,
		\end{enumerate}
		then, for any $l\in\rr$ such that $\mathcal{V}_l^{-1}\subset\mathcal{D}_G$, $\mathcal{V}_l^{-1}$ is positively invariant and is an estimation of the $\eta$-RoA. Moreover, if on $\mathcal{D}_G$, $\|z\|\to\infty$ implies $\mathcal{V}\to\infty$, then $z$ is bounded on $\mathcal{V}_l^{-1}$.
	\end{theorem}
	\begin{proof}
		Since $\mathcal{V}$ is continuous and non-increasing on $\mathcal{D}_G$, $\mathcal{V}_{l}^{-1}$ is positively invariant under \eqref{eq:system}. Hence, any solution starting in $\mathcal{V}_l^{-1}$ will stay in $\mathcal{D}_G$ for all $t\geq0$. Since $\mathcal{V}(x(t),z(t))$ is non-increasing and bounded from below by zero, $\mathcal{V}(x(t),z(t))$ must have a finite limit $v\geq0$ as $t\to\infty$.
		
		Now we prove $v=0$. If $v>0$, by 1) we have $0<v\leq\beta(\|\xi\|)$ for all $t\geq0$, which implies $\|\xi\|\geq\beta^{-1}(v)>0$. Hence, by 2) it holds that for all $t\geq0$
		\[\dot{\mathcal{V}}\leq-\gamma(\|\xi\|)\leq-\gamma(\beta^{-1}(v))<0.\]
		Hence, for all $t\geq0$
		\begin{equation*}
		\begin{aligned}
		\mathcal{V}(x(t),z(t))=&\mathcal{V}(x_0,z_0)+\int_{0}^{t}\dot{\mathcal{V}}(x(\tau),z(\tau))\mathrm{d}\tau\\
		\leq& \mathcal{V}(x_0,z_0)-\gamma(\beta^{-1}(v))t.
		\end{aligned}
		\end{equation*}
		The right-hand side will eventually become negative, which contradicts $v>0$. Hence, it holds that $0\leq\alpha(\|\eta\|)\leq \mathcal{V}\to0$, as $t\to\infty$. It follows from the sandwich theorem that $\eta\to0$, as $t\to\infty$, which proves the first claim.
		
		For the second claim, it directly follows that for every finite $l$, $z$ must be bounded on $\mathcal{V}_l^{-1}$, otherwise, $\mathcal{V}\to\infty$ and $l$ cannot be bounded.
	\end{proof}
	\begin{remark}
		In usual Lyapunov-like functions \cite{khalil2002nonlinear}, one would require the lower and upper bounds of $\mathcal{V}$ are $\mathcal{K}$ functions of the same variable, e.g., $\alpha(\|\eta\|)\leq \mathcal{V}\leq\beta(\|\eta\|)$, which can be viewed as a special case of Theorem \ref{th:s1} when setting $\xi=\eta$. In that case, Theorem \ref{th:s1} implies $\{(x,z)\in\mathcal{V}_l^{-1}|\eta(x,z)=0\}$ is a positively invariant set. In fact, given $(x_0,z_0)\in\mathcal{V}_l^{-1}$ and $\xi(x_0,z_0)=0$, since $\mathcal{V}$ is non-increasing, 1) implies that for all $t\geq0$,
		\[0\leq \mathcal{V}(x(t),z(t))\leq \mathcal{V}(x_0,z_0)\leq0,\]
		and hence $\{(x,z)\in\mathcal{V}_l^{-1}|\xi(x,z)=0\}$ is positively invariant.
		However, the set $\{(x,z)\in\mathcal{V}_l^{-1}|\eta(x,z)=0\}$ is often not positively invariant for power systems when $\eta=h$ or $\eta=f_2$. For example, consider the simple power system in Example \ref{ex:sg}, in which case the set $\{f_2=0\}=\{\omega=0,\dot{E}'_q=0\}$ is not invariant since on the set there exist points so that $\dot{\omega}\neq0$. Hence, to admit power systems applications we introduce $\xi$ as a relaxation. Note that by 1) $\xi=0$ implies $\eta=0$. Hence, $\{(x,z)\in\mathcal{V}_l^{-1}|\xi(x,z)=0\}\subset\{(x,z)\in\mathcal{V}_l^{-1}|\eta(x,z)=0\}$, which admits the former smaller set being positively invariant while the latter larger set not.
	\end{remark}
	\begin{remark}
		Theorem \ref{th:s1} is, to the best of our knowledge, a new result even though its formulation is based on similar concepts in the study of partial stability \cite{vorotnikov2012partial} and the pioneering work of J. L. Willems \cite{willems1974partial}. However, our result differs from partial stability, since we focus on the convergence of $\eta$ rather than part of state variables. It also differs from Willems's condition, as we do not require that $\mathcal{V}(x,z)$ vanish at the origin, which allows for equilibrium-independent analysis. Nevertheless, one important common feature among them is that they all differ from the classical Lyapunov function that requires $\mathcal{V}$ to be positive definite w.r.t. certain equilibrium. In fact, a function $\mathcal{V}(x,z)$ satisfying condition 1) in Theorem \ref{th:s1} is not necessarily a positive definite function in $\cal D$, or even in $\mathcal{D}_G$ since we allow it to vanish when $\eta(x,z)=0$. As will be illustrated in Section \ref{sec:case}, this property enables analysis without the prior knowledge of the targeted equilibrium.
	\end{remark}
	
	A natural candidate for a type-I $\mathcal{V}$-function is given in the Krasovskii's form as follows.
	\begin{equation}
	\mathcal{V}(x,z)=f(x,z)^TPf(x,z),
	\end{equation}
	where $P=P^T\succ0$ is a constant positive definite matrix. It meets the condition 1) in Theorem \ref{th:s1} with $\xi=f$, $\eta=f_2$, and
	\[\lambda_{min}(P)\|f_2\|^2\leq \mathcal{V}\leq\lambda_{max}(P)\|f\|^2,\]
	where $\lambda_{min}(P)$ and $\lambda_{max}(P)$ denote the minimal and the maximal eigenvalue of $P$, respectively.
	It yields
	\begin{equation*}
	\dot{\mathcal{V}}=\frac{\partial \mathcal{V}}{\partial x}\dot{x}+\frac{\partial \mathcal{V}}{\partial z}\dot{z}
	=f^T\left( PJ(x,z)+ J(x,z)^TP \right)f,
	\end{equation*}
	where
	\[J(x,z)=\frac{\partial f}{\partial x}-\frac{\partial f}{\partial z}\left( \frac{\partial g}{\partial z}\right) ^{-1}\frac{\partial g}{\partial x}.\]
	
	If one can verify that for all $(x,z)\in\mathcal{D}_G$
	\begin{equation}\label{eq:contraction}
	PJ(x,z)+ J(x,z)^TP\prec0,
	\end{equation}
	then $\mathcal{V}$ meets the condition 2) in Theorem \ref{th:s1}.
	
	In fact \eqref{eq:contraction} can be relaxed in case $\xi\neq f$. Section \ref{sec:case} will provide an example of using this type-I $\mathcal{V}$-function to estimate the $f_2$-RoA of a power system.
	
	\begin{remark}
		Regions satisfying condition \eqref{eq:contraction} is related to the concept of contracting region in the study of contraction analysis in DAEs systems \cite{nguyen2020contraction}. One may further allow $P$ to be state-dependent and the positive definite matrix $P(x,z)$ can serve as a metric of the space. 
	\end{remark}
	\subsection{Type II}
	In some cases, it may be difficult to obtain a sign-definite bound condition on $\mathcal{V}$ as required in Theorem \ref{th:s1}. It can be relaxed by restricting $\dot{\mathcal{V}}$ to employ the Barbalat's lemma.
	\begin{theorem}\label{th:s2}
		If there exist a scalar $C^1$ function $\mathcal{V}:\rr^n\times\rr^m\to\rr$, and a $\cal K$ function $\gamma$ such that:
		\begin{enumerate}
			\item $\mathcal{V}(x,z)$ is bounded from below on $\mathcal{D}_G$,
			\item $\dot{\mathcal{V}}(x,z)\leq-\gamma(\|\eta(x,z)\|)$ for every $(x,z)\in\mathcal{D}_G$,
			\item $\dot{\mathcal{V}}(x,z)$ is uniformly continuous in $t$,
		\end{enumerate}
		then, for any $l\in\rr$ such that $\mathcal{V}_l^{-1}\subset\mathcal{D}_G$, $\mathcal{V}_l^{-1}$ is positively invariant and is an estimation of the $\eta$-RoA. Moreover, if on $\mathcal{D}_G$, $\|z\|\to\infty$ implies $\mathcal{V}\to\infty$, then $z$ is bounded on $\mathcal{V}_l^{-1}$.
	\end{theorem}
	\begin{proof}
		By the same argument as in the proof of Theorem \ref{th:s1}, $\mathcal{V}_l^{-1}$ is positively invariant under \eqref{eq:system} and $\mathcal{V}(x(t),z(t))$ must have a finite limit as $t\to\infty$.
		Since $\dot{\mathcal{V}}$ is uniformly continuous in $t$, it follows from the Barbalat's Lemma\cite[Lemma 8.2]{khalil2002nonlinear} that $\dot{\mathcal{V}}\to0$ as $t\to\infty$. By 2) this implies $\eta\to0$ as $t\to\infty$, and hence $\mathcal{V}_l^{-1}$ is an estimation of the $\eta$-RoA. The second claim follows from the same argument as in the proof of Theorem \ref{th:s1}.
	\end{proof}
	
	Often, instead of 3), it is more convenient to verify a stronger condition that is $\ddot{\mathcal{V}}(x,z)$ being bounded along the solution. An example of using type-II $\mathcal{V}$-function will be present in Section \ref{sec:case}.
	
	\subsection{Type III}
	In some cases, it is only possible to construct a $\mathcal{V}(x,z)$ with a sign-constant (not sign-definite) derivative $\dot{\mathcal{V}}\leq0$. In this case, conditions of state convergence were obtained by LaSalle, Barbashin, and Krasovskii, by exploiting the properties of the $\omega$-limit set. Here, the same idea applies to $\eta$-RoA under an additional requirement of the set $\{\dot{\mathcal{V}}(x,z)=0\}$.
	\begin{theorem}\label{th:lasalle}
		Let $\mathcal{R}\subset\mathcal{D}_G$ be a compact and positively invariant set under \eqref{eq:system}. Let $\mathcal{V}:\rr^n\times\rr^m\to\rr$ be a $C^1$ function, and let $\Upsilon$ denote the largest invariant set of $\{(x,z)\in\mathcal{R}|\dot{\mathcal{V}}(x,z)=0\}$. If it holds that:
		\begin{enumerate}
			\item $\dot{\mathcal{V}}(x,z)\leq0$ for every $(x,z)\in\mathcal{R}$,
			\item $\Upsilon\subset\{(x,z)|\eta(x,z)=0\}$,
		\end{enumerate}
		then, $\mathcal{R}$ is an estimation of the $\eta$-RoA.
	\end{theorem}
	\begin{proof}
		It follows from LaSalle's Theorem\cite[Theorem 4.4]{khalil2002nonlinear} that every solution in $\mathcal{R}$ approaches $\Upsilon$ as $t\to\infty$. Since $\Upsilon\subset\{(x,z)|\eta(x,z)=0\}$, every solution approaches $\{(x,z)|\eta(x,z)=0\}$ as $t\to\infty$.
	\end{proof}
	An example of the type-III $\mathcal{V}$-function is the well-known energy functions for lossless multi-machine power systems with $\eta=f_2$ (see \cite{chiang1987foundations} for example).
	
	%\begin{remark}
	Unlike previous theorems, type-III does not require any bound condition of $\mathcal{V}(x,z)$. The closure of any bounded sublevel set of $\mathcal{V}$ is a valid candidate of the compact and positively invariant set $\mathcal{R}$, although the choice $\mathcal{R}$ does not have to be tied up with $\mathcal{V}$. For example, the closure of any bounded solution is a valid candidate of $\mathcal{R}$. Therefore, similar to the well-investigated energy function, the existence of a global type-III $\mathcal{V}$-function excludes all bounded complex behaviors of power systems, e.g., sustained oscillation, chaos, etc.
	%\end{remark}
	
	To end this section, we again emphasize that all three types of $\cal V$ functions do not require any knowledge of system equilibrium. And naturally, they concern the convergence of $\eta$ to zero rather than the convergence of all states to any pre-specified point. This distinguishes our theorems from the classical Lyapunov stability theory in terms of both what is requested and what is concluded. 
	%	\section{Exclude Unstable Equilibrium}
	%	The following Proposition provides a method to exclude unstable equilibria in $\mathcal{V}_l^{-1}$ resulting from type-I or type-II $\mathcal{V}$-functions.
	%	\begin{theorem}\label{th:unstable}
	%		Consider a type-I or II $\mathcal{V}$-function that is valid in $\mathcal{D}_G$ and there exists $\mathcal{V}_l^{-1}\subset\mathcal{D}_G$. If it holds that:
	%		\begin{enumerate}
	%			\item any $(x,z)\in\mathcal{E}\cap\mathcal{V}_l^{-1}$ is a local minimum of $\mathcal{V}$,
	%			\item if the solution $(x(t),z(t))$ of \eqref{eq:system} is not an equilibrium, the set $\{t\geq0|\dot{ \mathcal{V}}(x(t),z(t))=0\}$ has zero measure in $\rr$,
	%		\end{enumerate}
	%		then, $\mathcal{V}_l^{-1}$ contains no unstable equilibrium of \eqref{eq:system}.
	%	\end{theorem}
	%	An example of verifying these conditions for a power system will be given in Section \ref{sec:case}.
	\section{Illustrative Examples}\label{sec:case}
	This section presents two examples to illustrate potential applications of our results in power systems. In each example, we would employ previous theorems to first justify AS-detectability and then construct a $\mathcal{V}$ function. We begin with a single-machine-single-load system with a proportional-integral (PI) regulator, in which case a type-II $\mathcal{V}$-function is proposed that justifies the PI regulator can render the system to augmented synchronization from almost all initial points. In the second case, we consider the IEEE 9-bus system to demonstrate the modular property of AS-detectability and to show an estimation of $f_2$-RoA by a type-I $\mathcal{V}$-function. It is also an example of inherently non-isolated equilibria, in which case the classical point-based stability analysis cannot apply.

	\subsection{Single-Machine-Single-Load System}\label{subsec:PI}
	%	\subsubsection{Model}
	Consider the classical generator model connecting to a constant $PQ$ load via a lossless transmission line (Fig. \ref{fig:SMSL}).
	\begin{figure}[h]
		\centering
		\setlength{\abovecaptionskip}{0.2cm}	
		\setlength{\belowcaptionskip}{-0.2cm}
		\includegraphics[width=.6\hsize]{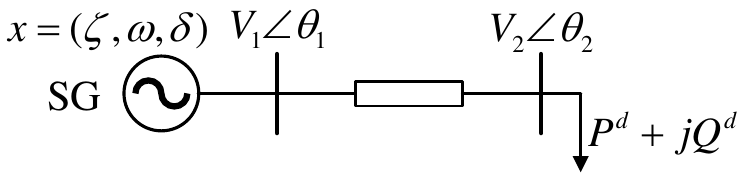}
		\caption{A synchronous generator connecting to a constant PQ load.}
		\label{fig:SMSL}
	\end{figure}
	
	The system is governed by \cite{1085625}
	\begin{equation}\label{eq:sg-2nd}
	\left\lbrace\begin{aligned}
	&\dot{\delta}=\omega\\
	&M\dot{\omega}=-D\omega-P^e+P^g,
	\end{aligned} \right.
	\end{equation}
	together with $P^g=P^g_0+u$,
	%	\[P^g=P^g_0+u\]
	where $P^g_0$ is the fixed mechanical power input and $u$ is a simple PI regulator for some $k_1>0$ and $k_2>0$ defined as
	\[u(t)=-k_1\omega(t)-k_2\int_{0}^{t}\omega(\tau)\mathrm{d}\tau,\]
	or equivalently in the differential form by introducing state $\zeta$
	\begin{equation*}
	\dot{\zeta}=-k_2\omega,\quad u=-k_1\omega+\zeta.
	\end{equation*}
	The terminal output power reads
	\begin{equation*}
	P^e=\frac{EV_1\sin(\delta-\theta_1)}{x'_{d}},Q^e=\frac{EV_1\cos(\delta-\theta_1)}{x'_{d}}-\frac{1}{x'_{d}}V_1^2,
	\end{equation*}
	where $E>0$ is the constant generator internal voltage. The algebraic equations read
	\begin{equation*}
	\left\lbrace \begin{aligned}
	0&=P^e-B_{12}V_1V_2\sin\theta_{12}\\
	0&=Q^e+B_{11}V_1^2+B_{12}V_1V_2\cos\theta_{12}\\
	0&=P^d-B_{12}V_1V_2\sin\theta_{21}\\
	0&=Q^d+B_{22}V_2^2+B_{12}V_1V_2\cos\theta_{12},
	\end{aligned}\right.
	\end{equation*}
	where $V_1\angle\theta_1$ and $V_2\angle\theta_2$ are complex voltages at buses 1 and 2, respectively.
	
	%	\subsubsection{Analysis}
	We first show this system satisfies conditions of Theorem \ref{th:auto2} and hence any non-degenerate solution is AS-detectable. By the same argument in the proof of Proposition \ref{pro:sg}, bounded $\delta(t)$ implies bounded $\omega(t)$. Since $\zeta(t)=\zeta(0)-k_2\int_{0}^{t}\omega(\tau)\mathrm{d}\tau=\zeta(0)-k_2(\delta(t)-\delta(0))$,  bounded $\delta(t)$ also implies bounded $\zeta(t)$. Hence, condition 1) of Theorem \ref{th:auto2} is met. For any solution satisfying $f_2=\omega\equiv0$, it follows $\dot{\omega}\equiv0$ and $\dot{\zeta}\equiv0$, which meets condition 2) of Theorem \ref{th:auto2}, and hence any non-degenerate solution of this example is AS-detectable.
	
	Next, we estimate the $f_2$-RoA of the system by construct a type-II $\mathcal{V}$-function. Let
	\begin{equation*}
	\begin{aligned}
	\mathcal{Q}(x,z)=&\frac{1}{2}M\omega^2+\frac{EV_1(1-\cos(\delta-\theta_1))}{x'_d}\\
	&+B_{12}V_1V_2(1-\cos\theta_{12}),
	\end{aligned}
	\end{equation*}
	Since $B_{12}>0$ and $E>0$, the function $\mathcal{Q}(x,z)$ is bounded from below on $\mathcal{D}_G$, provided $V_1>0$, $V_2>0$.
	
	In this case, $x_1=\text{col}(\zeta,\omega)$, $x_2=\delta$, and $z=\text{col}(\theta_1,V_1,\theta_2,V_2)$. Direct calculation yields
	\[\left(\frac{\partial g}{\partial z}\right)^{-1}\frac{\partial g}{\partial x_2}=\begin{bmatrix}
	-1&0&-1&0
	\end{bmatrix}^T.\]
	Hence, by \eqref{eq:zdot}, it holds $\dot{z}=\text{col}(\omega,0,\omega,0)$. It holds
	%	 $\dot{\mathcal{Q}}(x,z)=-D\omega^2+\omega(P_0^g+P^d+u)$.
	\[\dot{\mathcal{Q}}(x,z)=-D\omega^2+\omega(P^d+P_0^g+u).\]
	Consider the following $\mathcal{V}$-function
	\[\mathcal{V}(x,z)=\mathcal{Q}(x,z)+\frac{1}{2k_2}(\zeta+P^d+P_0^g)^2.\]
	Clearly, $\mathcal{V}(x,z)$ is lower bounded. And it holds
	\[\frac{\mathrm{d}}{\mathrm{d}t}\frac{1}{2k_2}(\zeta+P^d+P_0^g)^2=-\omega(u+P^d+P_0^g)-k_1\omega^2.\]
	Hence,
	\[\dot{\mathcal{V}}(x,z)=-(D+k_1)\omega^2=-(D+k_1)\|f_2\|^2,\]
	which verifies condition 2) in Theorem \ref{th:s2}.
	To verify condition 3), we prove $\ddot{\mathcal{V}}$ is bounded when $\mathcal{V}$ is bounded. Direct calculation yields
	\begin{equation*}
	\begin{aligned}
	\ddot{\mathcal{V}}(x,z)&=-2(D+k_1)\omega\dot{\omega}\\
	&=-\frac{2(D+k_1)}{m}\omega(-D\omega+P^g_0+u-P^e).
	\end{aligned}
	\end{equation*}
	It is easy to show by contradiction arguments that $\mathcal{V}$ being bounded implies $\omega$ and $u$ being bounded. Note also $|P^e|\leq EV_1/x'_d$. Hence, $\ddot{\mathcal{V}}$ is bounded in any sublevel set $\mathcal{V}_l^{-1}$ with finite level value $l$. Hence, any $\mathcal{V}_l^{-1}\subset\mathcal{D}_G$ is a valid estimation of $f_2$-RoA of the system.
	
	%	\subsubsection{Implications}
	For illustration, this example was built on a single-machine system and a synchronous generator. However, we remark that the analysis applies to large systems and similar dynamical subsystems, e.g., the virtual synchronous generator with energy storage. Although the generator's primary power input is conventionally assumed to be constant in transient stability analysis since the governor dynamics is usually slow, this example shows such an assumption can be relaxed, which might be necessary for microgrids with inverter-interfaced fast response power resources.
	\subsection{IEEE 9-bus System}\label{subsec:9bus}
	We now consider the IEEE 9-bus system as shown in Fig.\ref{fig:9bus}. To demonstrate the capability of our results in heterogeneous settings, we modify the benchmark so that it consists of two different synchronous generators, an inverter-interfaced power source, and three constant PQ loads. We assume the generator at bus 1 has a strong excitation control and hence is described by \eqref{eq:sg-2nd} together with the PI regulator. The generator at bus 2 is governed by the flux-decay model as in Example \ref{ex:sg}. And the power source at bus 3 is modeled as in Example \ref{ex:ps}. The differential equations of each subsystem can be found in previous examples.
	\begin{figure}[h]
		\centering
		\setlength{\abovecaptionskip}{0.2cm}	
		\setlength{\belowcaptionskip}{-0.2cm}
		\includegraphics[width=.8\hsize]{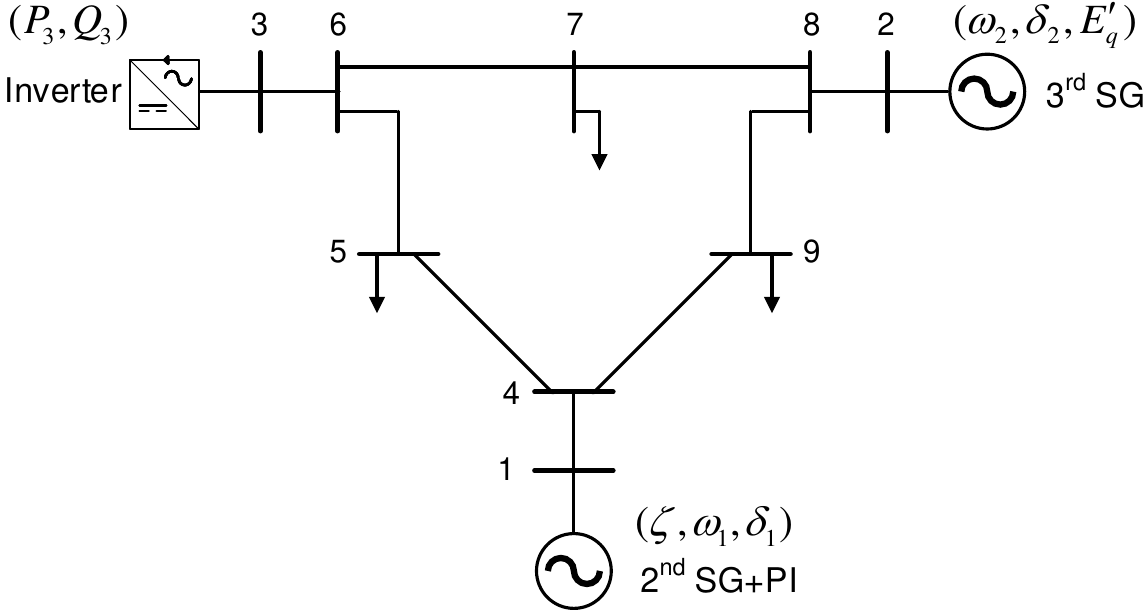}
		\caption{The IEEE 9-bus system consisting of the three aforementioned subsystems and three constant PQ loads.}
		\label{fig:9bus}
	\end{figure}
	
	The algebraic equations consist of power balance constraints at each bus, which read for $i=1,\ldots,9$
	\begin{equation*}
	\left\lbrace \begin{aligned}
	P_i&=G_{ii}V_i^2+\sum_{j\in\mathcal{N}_i}V_iV_j(B_{ij}\sin\theta_{ij}+G_{ij}\cos\theta_{ij})\\
	Q_i&=-B_{ii}V_i^2-\sum_{j\in\mathcal{N}_i}V_iV_j(B_{ij}\cos\theta_{ij}-G_{ij}\sin\theta_{ij}).
	\end{aligned}\right. 
	\end{equation*}
	
	Network and load parameters were obtained from the MATPOWER package \cite{zimmerman2010matpower}. Table \ref{tab1} reports parameters of three dynamical subsystems.
	\begin{table}[h]
		\centering
		\footnotesize
		\caption{Parameters of 3 subsystems.}
		\label{tab1}
		\begin{tabular}{l|l|l}
			\hline
			Bus&\multicolumn{1}{c|}{Parameters}&\multicolumn{1}{c}{Values}\\
			\hline
			\multirow{2}*{1}& $M$, $D$, $E$, $x_{d}'$,& 0.075, 0.032, 1.01, 0.061,\\~ &$P^g_0$ $k_1$, $k_2$  & 0.02, 0.10, 0.72 \\ 
			\hline
			\multirow{2}*{2}& $M$, $D$, $T_{d0}'$, $x_{q}$,&0.02, 0.003, 1.00, 0.20,\\ &$x_{d}$, $x_{d}'$, $P^g$, $E^f$ &  0.896, 0.12, 1.63, 1.52 \\ 
			\hline
			\multirow{2}*{3}& $\tau_{1}$, $\tau_{2}$, $d_{1}$, $d_{2}$, &10.00,   10.00,    0.10,   0.10,  \\
			& $P^{\text{ref}}$, $ Q^{\text{ref}}$, $\theta^{\text{ref}}$, $V^{\text{ref}}$ & 0.85, -0.0365,    0.0833,    1.00 \\
			\hline
		\end{tabular}
	\end{table}
	
	\subsubsection{Non-isolated and not asymptotically stable equilibrium}	
	The collective state variables of the system read 
	\[x=\text{col}(\zeta,\omega_1,\delta_1,\omega_2,\delta_2,E'_q,P_3,Q_3).\]
	The collective algebraic variable $z$ is of 18 dimension and consists of $\theta_i$ and $V_i$ at each bus.
	Setting $\zeta=0$, one equilibrium can be obtained, which reads 
	\[x^*_0=\text{col}(0,0,0.0431,0,0.4756,1.0288,0.8500,-0.0365).\]
	However, $x^*_0$ is not an isolated equilibrium. In fact, there is a continuous equilibria trajectory,	along which it always holds that $\omega_1=\omega_2=0$ but other states keep varying. Fig.\ref{fig:nonisolate} shows this equilibria trajectory as a function of $\zeta$.
	\begin{figure}[h]
		\centering
		\setlength{\abovecaptionskip}{0.2cm}	
		\setlength{\belowcaptionskip}{-0.2cm}
		\includegraphics[width=1\hsize]{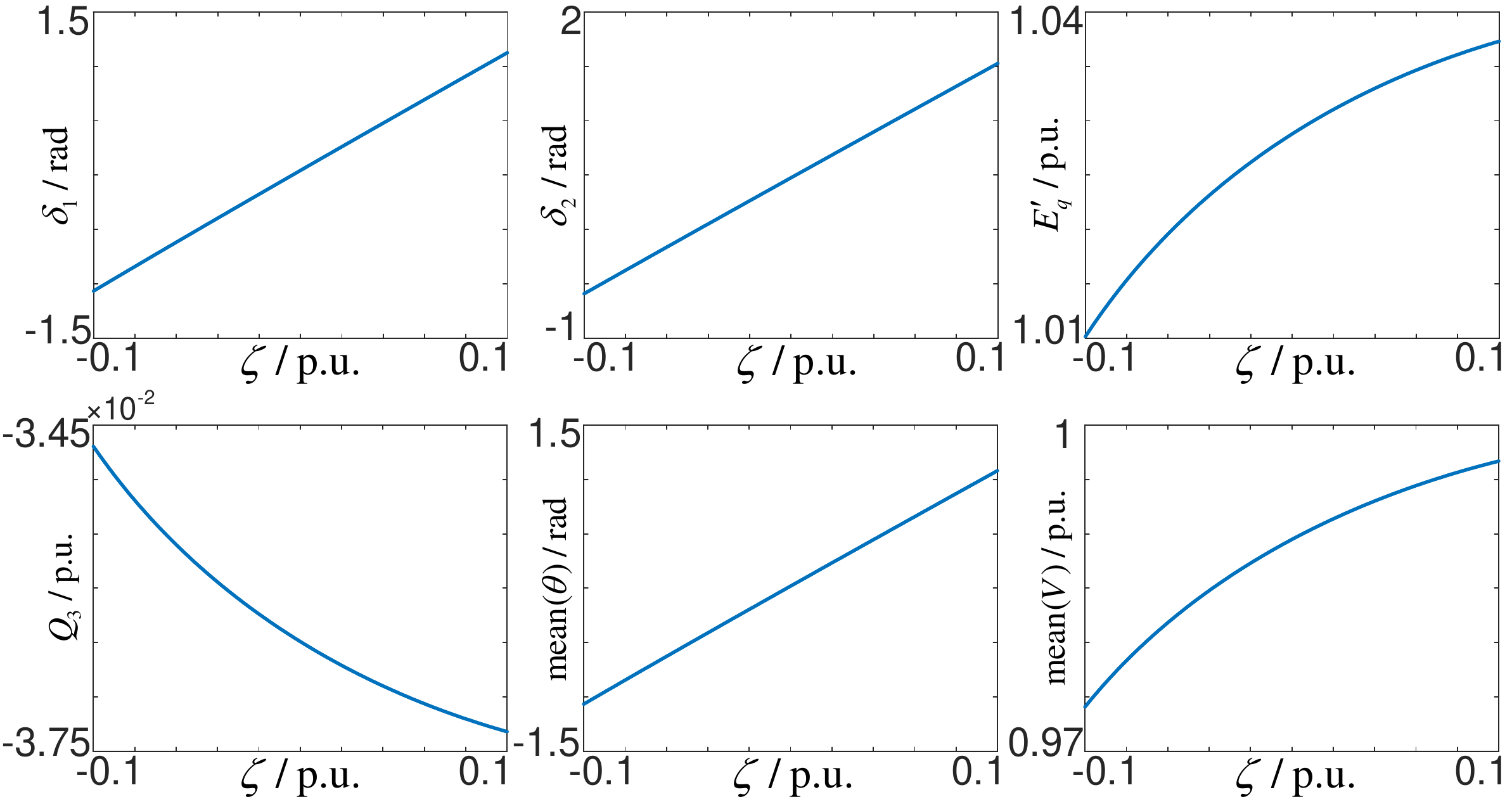}
		\caption{Equilibria continuum shown as a function of $\zeta$. Here mean($\theta$) and mean($V$) stand for the average value of all $\theta$ and $V$, respectively.}
		\label{fig:nonisolate}
	\end{figure}
	
	Unlike the well-investigated phase rotational symmetry, the equilibria continuum here cannot be eliminated by simple coordinate transformation. Indeed, phase rotational symmetry does not exist in this example due to the power-angle droop at bus 3. The inherently non-isolated equilibria actually result from the interaction between the PI regulator at bus 1 and the power-angle droop at bus 3.
	
	More importantly, any single equilibrium in this continuum is not Lyapunov asymptotically stable since any small perturbation along the direction tangent to the equilibria trajectory will make the system leave the initial equilibrium and settle in a new one. 
	%	For example, consider the equilibrium $x^*_0$. The system Jacobian at $x^*_0$ has a zero eigenvalue corresponding to the eigenvector
	%	$(0.0655,   0,    0.7152,   0,    0.6924, 0.0079,   -0.0690,   -0.0008)^T$, which is aligned with the equilibria trajectory.
	To illustrate, consider the projection of the system vector field onto the plane spanned by two orthogonal directions $s_1$ and $s_2$ around the equilibrium $x^*_0$. Here, we set $s_1=\text{col}(1,0,1,0,1,1,1,1)/\sqrt{6}$ and $s_2=\text{col}(0,1,0,1,0,0,0,0)/\sqrt{2}$. Hence, $s_2$ represents the variation of $\omega_1$ and $\omega_2$, while $s_1$ represents the combination of the others. Fig.\ref{fig:vector} reports the projected vector filed. It shows that all vectors point to the line $s_2=0$, but they do not point to any particular point. Note the line $s_2=0$ is the projection of the equilibria trajectory on this plane. Now suppose the system is subjected to a small disturbance that shifts the state from $x_0^*$ to $\tilde{x}_0$, which is represented by the red dot. From this initial point, the system solution can reach augmented synchronization but it converges to a different equilibrium $\tilde{x}^*$, which is represented by the green dot. This indicates that this equilibria trajectory as a whole has certain stability property but no single equilibrium is Lyapunov asymptotically stable.
	\begin{figure}[h]
		\centering
		\setlength{\abovecaptionskip}{0.2cm}	
		\setlength{\belowcaptionskip}{-0.2cm}
		\includegraphics[width=1\hsize]{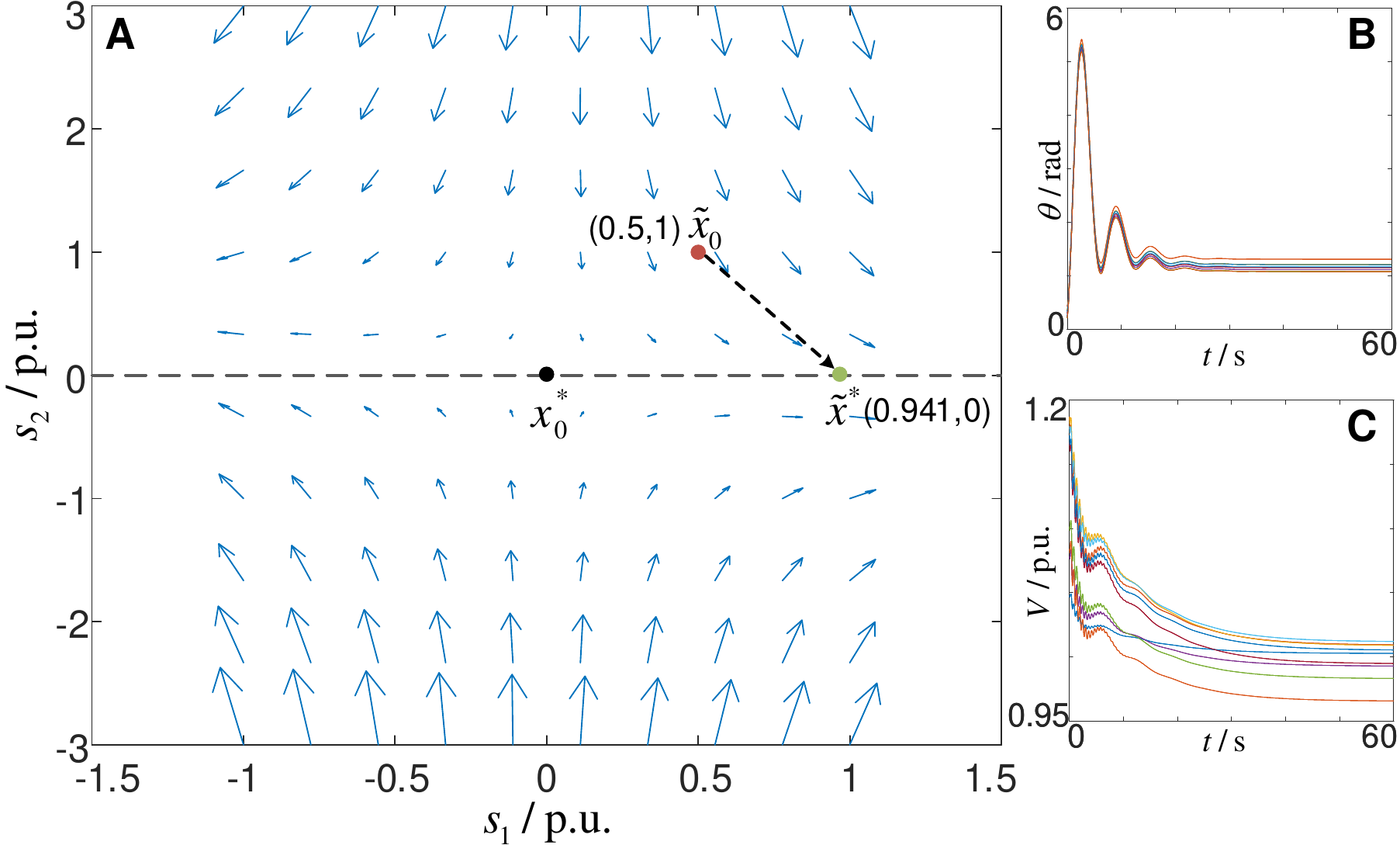}
		\caption{(A) The projection of the system vector field onto the plane spanned by $s_1$ and $s_2$ around $x^*_0$. (B,C) The solutions of $\theta$ and $V$ starting from $\tilde{x}_0$.}
		\label{fig:vector}
	\end{figure}

	\subsubsection{ AS-detectability}	
	This 9-bus system provides an example of verifying AS-detectability in a modular manner. The system has exactly the modular structure as defined in Section \ref{subsec:modular}. And each subsystem satisfies conditions 1) and 2) in Theorem \ref{pro:modular}, which can be shown by the same argument as in Example \ref{ex:sg}, \ref{ex:ps} and in Section \ref{subsec:PI}. Therefore, any non-degenerate solution of this system is AS-detectable.
	
	In this case, although $\omega(t)$ and $\dot{z}(t)$ may identically converge to 0, most other states such as $\zeta(t)$, $\theta(t)$, and $V(t)$ may converge to different values provided different initial points. Hence, one cannot single out any pre-specified equilibrium to perform transient stability analysis. Our results provide an alternative way to handle this situation by checking the convergence of $\dot{z}(t)$ instead of all states. 
	
	\subsubsection{ Estimation of the $f_2$-RoA}	
	To illustrate, let us estimate the $f_2$-RoA of the system by a type-I $\mathcal{V}$-function in the Krasovskii's form. For some positive definite matrix $P$, let $\mathcal{V}(x,z)=f(x,z)^TPf(x,z)$.
	%	\begin{equation*}
	%	\mathcal{V}(x,z)=f(x,z)^TPf(x,z)
	%	\end{equation*}
	Note for subsystem 1, $\dot{\zeta}$ is linearly related to $\dot{\delta}_1$. Hence, instead of letting
	\[\xi=f=\text{col}(\dot{\zeta},\dot{\omega}_1,\dot{\delta}_1,\dot{\omega}_2,\dot{\delta}_2,\dot{E}_q',\dot{P}_3,\dot{Q}_3),\]
	one can choose 
	\[\xi=\text{col}(\dot{\omega}_1,\dot{\delta}_1,\dot{\omega}_2,\dot{\delta}_2,\dot{E}_q',\dot{P}_3,\dot{Q}_3).\]
	This yields $f=A\xi$, where $A$ is a constant matrix and reads
	\[A=\begin{bmatrix}
	0&-k_2&0_{1\times5}\\&I_{7\times7}&
	\end{bmatrix}.\]
	Hence, condition \eqref{eq:contraction} can be relaxed as
	%	 $A^T(PJ(x,z)+ J(x,z)^TP)A\prec0$.
	\begin{equation}\label{eq:LMI}
	A^T(PJ(x,z)+ J(x,z)^TP)A\prec0.
	\end{equation}

	Numerical LMI solver yields a candidate of $P$ as reported in \eqref{eq:P} such that \eqref{eq:LMI} holds in $\mathcal{V}_l^{-1}$ with $l=4$.
	\begin{figure*}
		\begin{equation}\label{eq:P}
		\footnotesize	
		P=\begin{bmatrix}
		1.136	&	-0.049	&	0.026	&	0.225	&	-0.049	&	-0.050	&	-0.507	&	-0.066	\\
		-0.049	&	0.719	&	-0.151	&	-0.237	&	-0.041	&	-0.077	&	0.005	&	0.052	\\
		0.026	&	-0.151	&	0.497	&	-0.060	&	0.131	&	0.350	&	-0.053	&	0.174	\\
		0.225	&	-0.237	&	-0.060	&	1.070	&	-0.146	&	0.096	&	0.006	&	-0.047	\\
		-0.049	&	-0.041	&	0.131	&	-0.146	&	0.461	&	-0.302	&	0.073	&	-0.143	\\
		-0.050	&	-0.077	&	0.350	&	0.096	&	-0.302	&	0.651	&	-0.012	&	0.304	\\
		-0.507	&	0.005	&	-0.053	&	0.006	&	0.073	&	-0.012	&	0.530	&	0.038	\\
		-0.066	&	0.052	&	0.174	&	-0.047	&	-0.143	&	0.304	&	0.038	&	0.463	
		\end{bmatrix}
		\end{equation}
	\end{figure*}
	Then, a conservative estimation of the $f_2$-RoA is given by $\mathcal{V}_{4}^{-1}$.
	%	To verify condition 1) in Theorem \ref{th:unstable}, observe that at an equilibrium $f=0$ and hence $\mathcal{V}=0$, which is the global minimum. To verify condition 2), observe that $\dot{\mathcal{V}}=0$ is equivalent to $\xi=0$ and hence is equivalent to $f=0$. Along non-equilibrium solutions, $\{t\geq0|f(t)=0\}$ must have zero measure in $\rr$, otherwise since $f$ is continuously differentiable, there exit $0<t_1<t_2$ such that $f(t)=0$ for all $t\in(t_1,t_2)$, which implies the solution is an equilibrium. Hence, condition 2) in Theorem \ref{th:unstable} is verified. Therefore, no unstable equilibrium locates in $\mathcal{V}_{4}^{-1}$.
	
	In this case, $\mathcal{V}_{4}^{-1}$ is an 8-dimensional manifold in the 26-dimensional Euclidean space. To illustrate, we project it onto the $\zeta-\omega_1$ plane as shown in Fig.\ref{fig:phase9bus}. Two trajectories starting from different initial points were also projected onto the plane, which clearly shows the dependency of the converging point of $\zeta$ on the initial point. It also shows the independence of the converging point of $\omega_1$, which should always converge to 0.
	
	Fig.\ref{fig:traj1} and \ref{fig:traj2} give the solutions $\theta(t)$, $V(t)$ of the corresponding two trajectories, respectively. In both cases, the system approaches to the augmented synchronization, i.e., $\dot{z}(t)\to0$ as $t\to\infty$. However, the converging points of $\theta(t)$ and $V(t)$ are different. This again shows that augmented synchronization could be more desirable to power system transient stability analysis than equilibrium.

	\begin{figure}[h]
		\centering
		\setlength{\abovecaptionskip}{0.2cm}	
		\setlength{\belowcaptionskip}{-0.2cm}
		\includegraphics[width=.75\hsize]{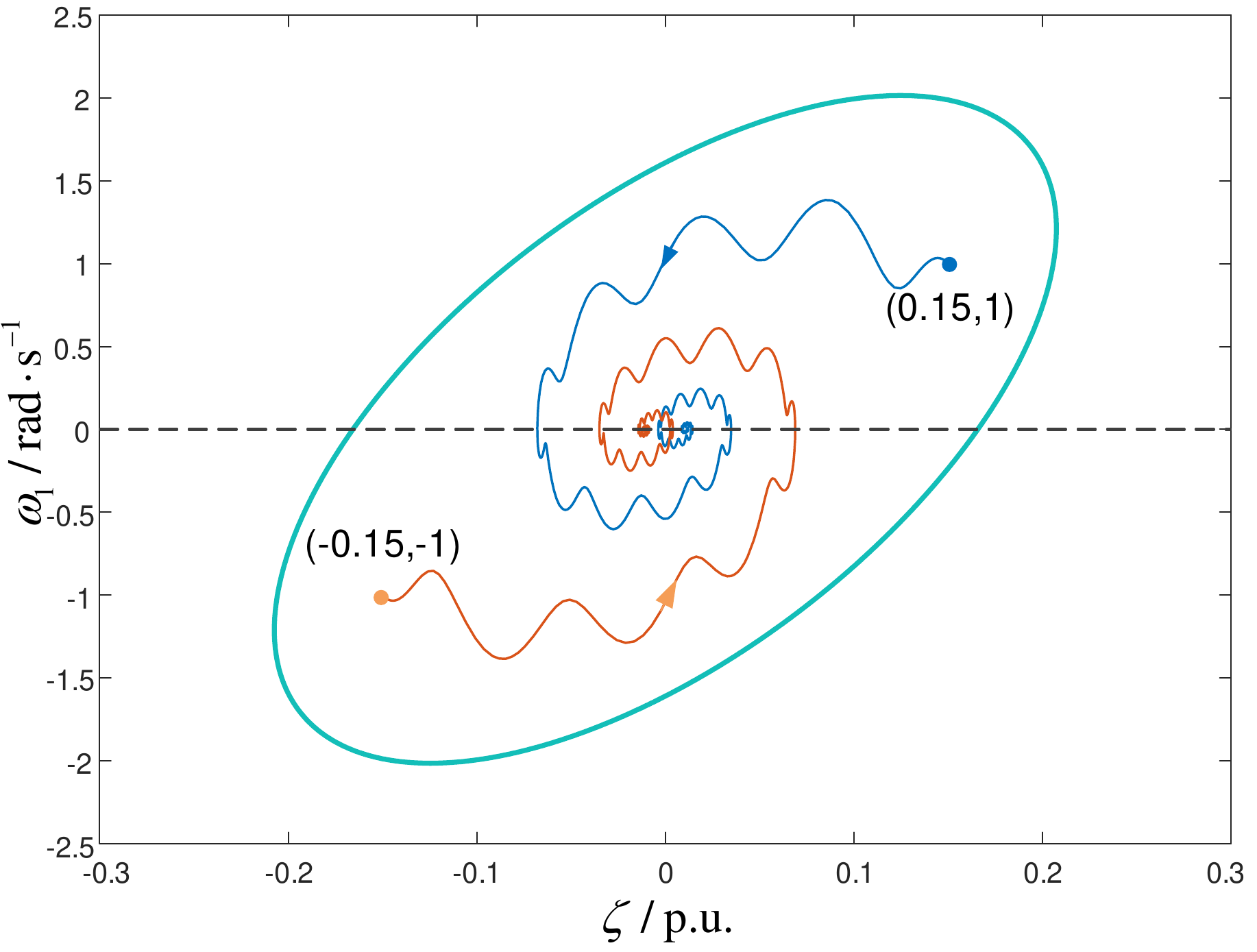}
		\caption{Projections of the estimated RoA and two trajectories. The green ellipse represents the boundary of $\mathcal{V}^{-1}_{4G}$. The orange an the blue curves are projections of two trajectories that stars from different initial points. Both trajectories converge to $w_1=0$ but to different $\zeta$.}
		\label{fig:phase9bus}
	\end{figure}
	
	\begin{figure}[h]
		\centering
		\setlength{\abovecaptionskip}{0.2cm}	
		\setlength{\belowcaptionskip}{-0.2cm}
		\includegraphics[width=.85\hsize]{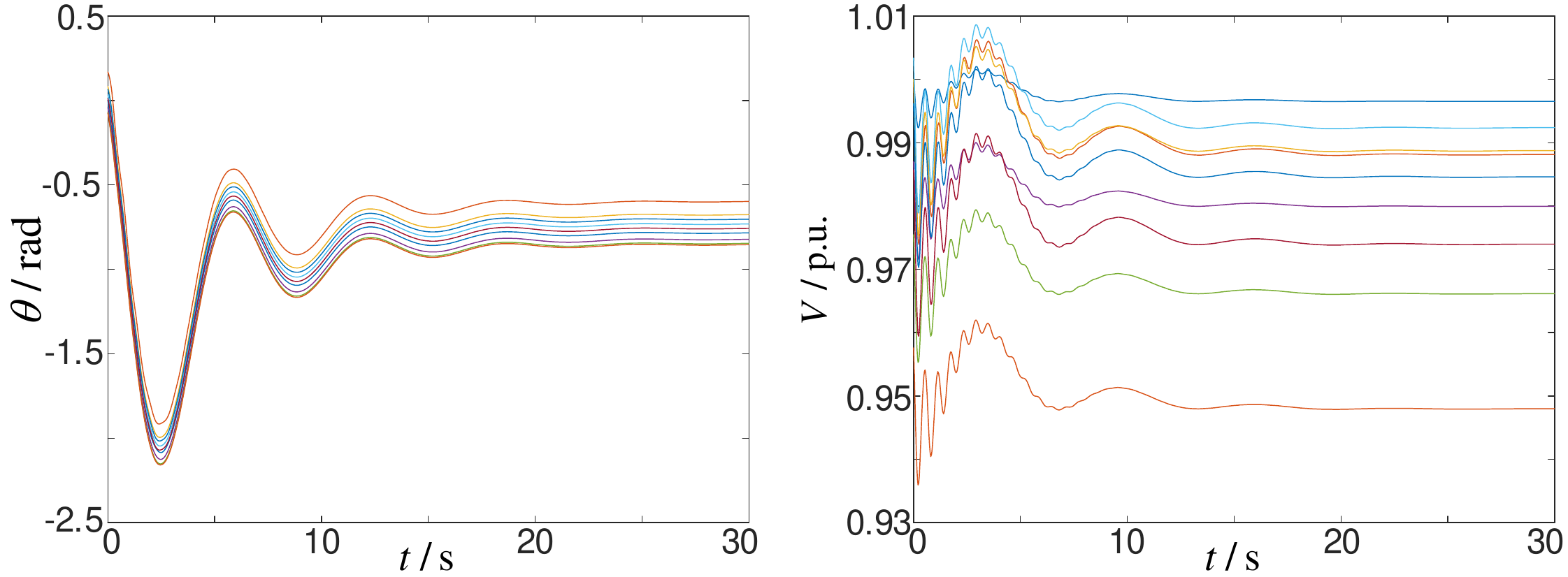}
		\caption{Solutions of $\theta$ and $V$ of all buses corresponding to the orange trajectory in Fig. \ref{fig:phase9bus}. System reaches augmented synchronization.}
		\label{fig:traj1}
	\end{figure}
	
	\begin{figure}[h]
		\centering
		\setlength{\abovecaptionskip}{0.2cm}	
		\setlength{\belowcaptionskip}{-0.2cm}
		\includegraphics[width=.85\hsize]{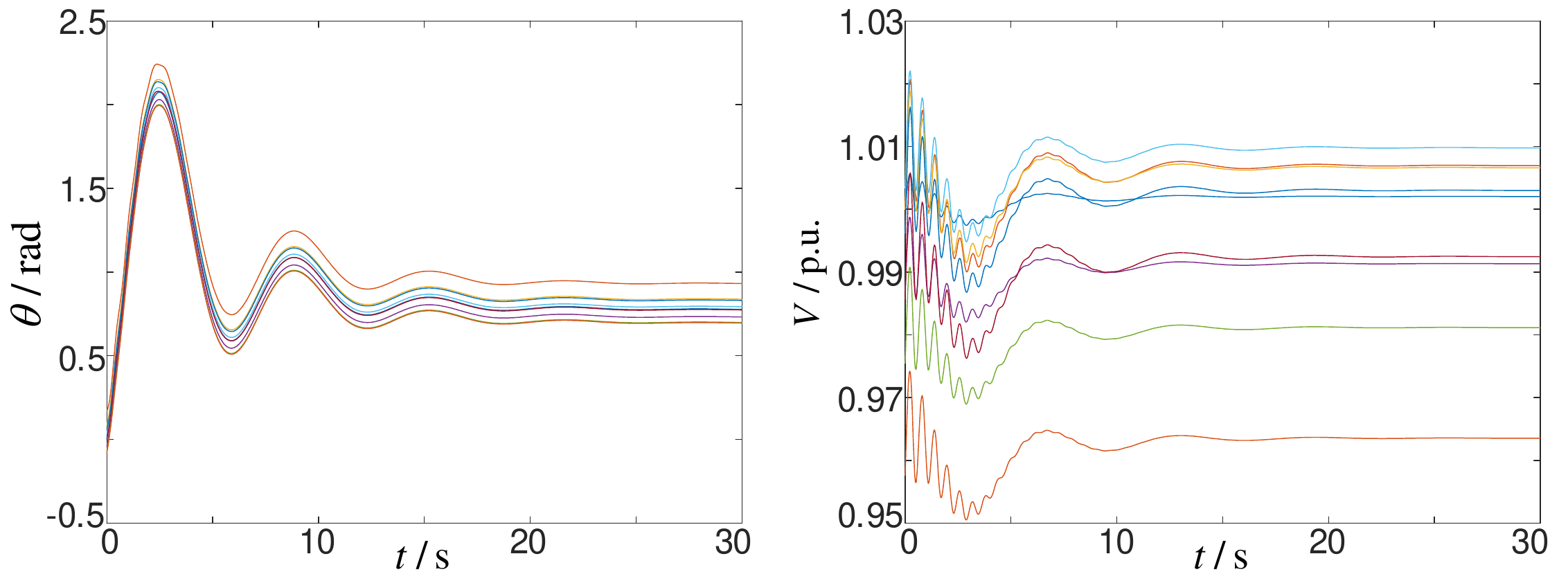}
		\caption{Solutions of $\theta$ and $V$ of all buses corresponding to the blue trajectory in Fig. \ref{fig:phase9bus}. System also reaches augmented synchronization, however, compared with Fig.\ref{fig:traj1}, $\theta$ and $V$ converge to different values.}
		\label{fig:traj2}
	\end{figure}
	\section{Concluding Remarks}
	In this paper, we re-defined the ``nominal motion'' of power systems by introducing the concept of augmented synchronization. We derived conditions for augmented synchronization detectability and presented a compositional approach to verify these conditions in modular structured power systems. That provides the long-standing practice wisdom with a solid theoretical foundation. Inspired by such wisdom, we further developed theorems to characterize augmented synchronization with different types of $\mathcal{V}$-functions. These theorems extend the classical Lyapunov-based direct method and can provide estimations of RoA w.r.t. augmented synchronization rather than an equilibrium.
	
	Our results shed new light on power system stability analysis, which may open the possibility for an equilibrium-independent analytic that better fits the demand of future smart grids. As an initial step, the concept of augmented synchronization may also provide a better perspective to understand power system stability by rethinking what should be the cornerstone, i.e., the ``nominal motion" in stability analysis. Our ongoing works include developing efficient algorithms to calculate $\eta$-RoA for large-scale power systems. 
	\bibliographystyle{IEEEtran}
	\bibliography{mybib}
	
	\appendices
	\makeatletter
	\@addtoreset{equation}{section}
	\@addtoreset{theorem}{section}
	\makeatother

\begin{IEEEbiography}[{\includegraphics[width=1in,height=1.25in,clip,keepaspectratio]{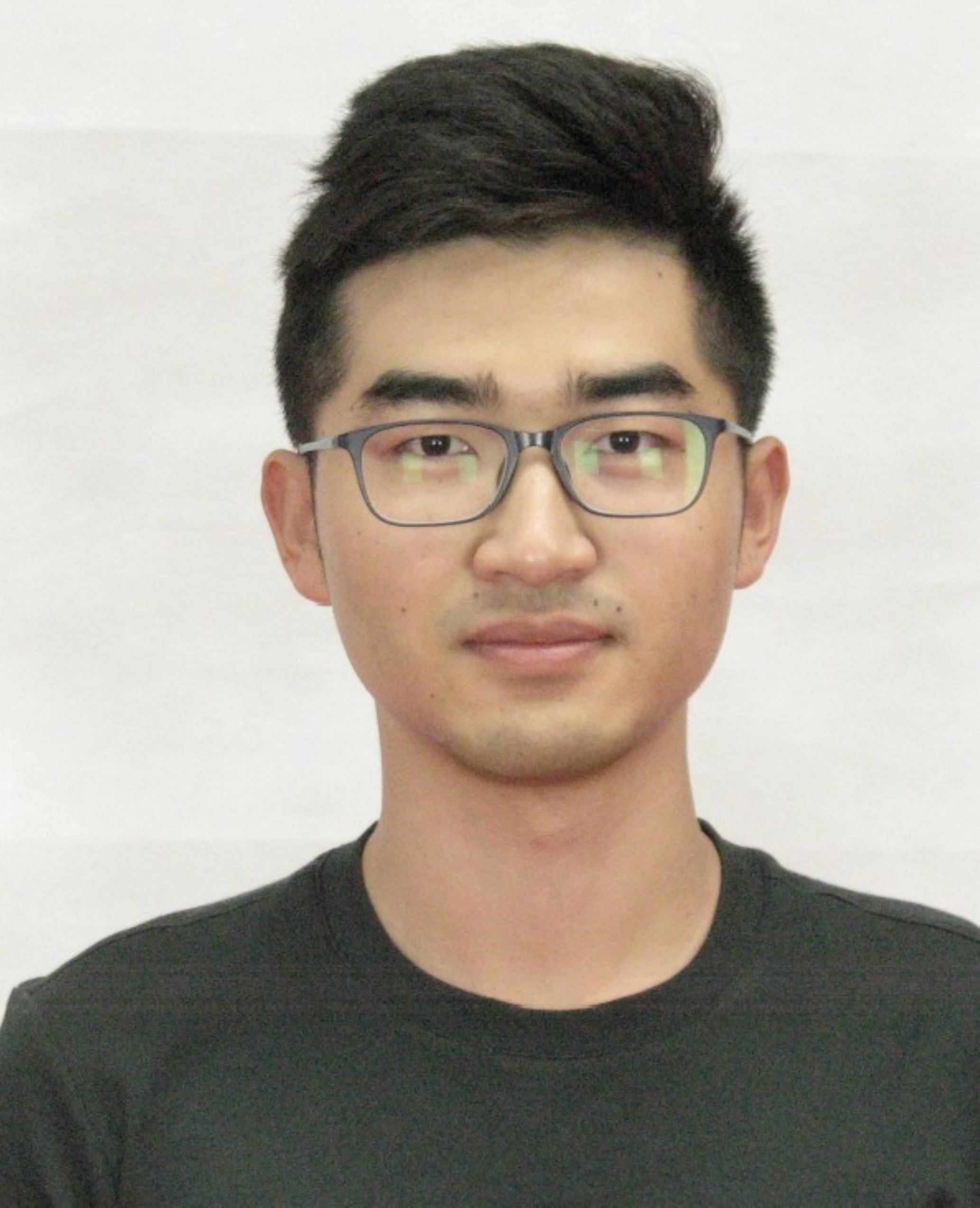}}]{Peng Yang}
	received the B.Sc. degree in electrical engineering and the B.Sc. degree in mathematics from Tsinghua University, Beijing, China, in 2017 and 2018, respectively. He is currently pursuing the Ph.D. degree in electrical engineering at Tsinghua University, Beijing, China. His research interests include power system stability analysis and control. He was the winner of 2020 Zhang Si-Ying (CCDC) Outstanding Young Paper Award. He received the Best Paper Award of IEEE Transactions on Power Systems in 2020.
\end{IEEEbiography}

\begin{IEEEbiography}[{\includegraphics[width=1in,height=1.25in,clip,keepaspectratio]{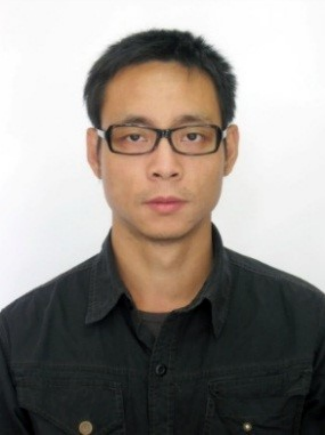}}]{Feng Liu}	
	(M'10, SM'18) received the B.Sc. and Ph.D. degrees in electrical engineering from Tsinghua University, Beijing, China, in 1999 and 2004, respectively. 
	
	Dr. Liu is currently an Associate Professor of Tsinghua University. From 2015 to 2016, he was a visiting associate at California Institute of Technology, CA, USA. His research interests include power system stability analysis, optimal control and robust dispatch, game theory and learning theory and their applications to smart grids. He is the author/coauthor of more than 00 peer-reviewed technical papers and three books, and holds more than 20 issued/pending patents. He is an associated editor of several international journals including IEEE Transactions on Smart Grid and Control Engineering Practice. He also served as a guest editor of IEEE Transactions on Energy Conversion.
\end{IEEEbiography}

\begin{IEEEbiography}[{\includegraphics[width=1in,height=1.25in,clip,keepaspectratio]{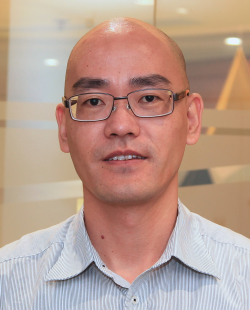}}]{Tao Liu} (M'13) received his B.E. degree from	Northeastern University, China, in 2003 and PhD degree from the Australian National University (ANU), Australia, in 2011. From 2012 to 2015, he worked as a Post-doctoral Fellow at ANU, University of Groningen, and University of Hong Kong (HKU). He became a Research Assistant Professor at HKU in 2015 and now is an Assistant Professor. His	research interests include power system analysis and	control, complex dynamical networks, distributed	control, and event-triggered control.
\end{IEEEbiography}

\begin{IEEEbiography}[{\includegraphics[width=1in,height=1.25in,clip,keepaspectratio]{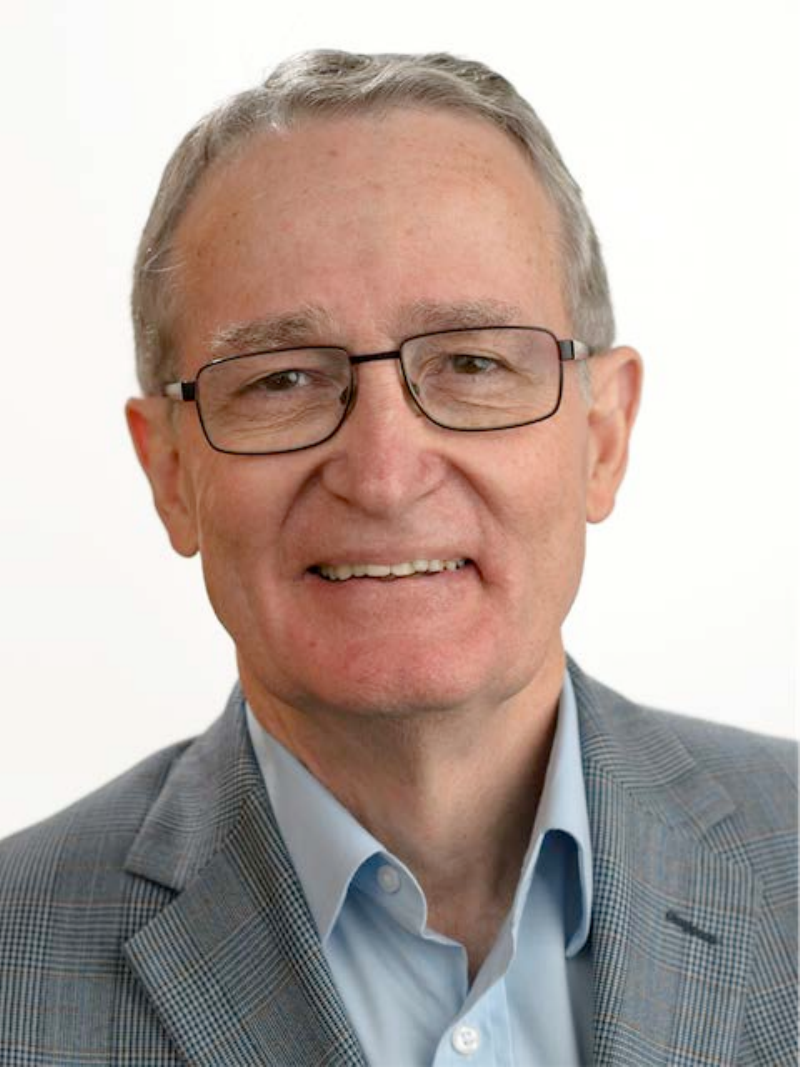}}]{David J. Hill} (S’72-M’76-SM’91-F’93-LF’14) received the PhD degree in Electrical Engineering from the University of Newcastle, Australia, in 1976. From 2021, he is a Professor in the School of Electrical Engineering and Telecommunications, The University of New South Wales, Sydney, Australia (UNSW Sydney) and a Program Leader in The Reliable Affordable Clean Energy for 2030 Cooperative Research Centre (RACE for 2030 CRC) where resources of about A\$350million have been obtained from industry and Government. He is also Professor Emeritus at The University of Sydney and The University of Hong Kong.
	
	During 2013-2020, he held the positions of Chair of Electrical Engineering and Director of the Centre for Electrical Energy Systems in the Department of Electrical and Electronic Engineering at the University of Hong Kong. He previously held positions at the University of Sydney including the Chair of Electrical Engineering during 1994-2002 and again in 2010-2013 along with an Australian Research Council Professorial Fellowship. He was Foundation Director of the Centre for Future Energy Networks during 2010-2018 and part-time Professor 2013-2020. During 2005-2010, he was an ARC Federation Fellow at the Australian National University. He has also held academic and substantial visiting positions at the universities of Melbourne, California (Berkeley), Newcastle (Australia), Lund (Sweden), Munich and in Hong Kong (City and Polytechnic Universities). He holds several honorary positions in Australia, Hong Kong and China. He is also a consultant in the area of power and energy issues in Australia and internationally.
	
	His research activities have been in energy systems, control systems, complex networks, learning systems and stability analysis. His work is now mainly focussed on issues for future energy and power and energy networks with the aim to bring science to accelerate the clean energy transition.
	
	Professor Hill is a Fellow of the Society for Industrial and Applied Mathematics, USA, the International Federation of Automatic Control, the Australian Academy of Science, the Australian Academy of Technological Sciences and Engineering and the Hong Kong Academy of Engineering Sciences. He is also a Foreign Member of the Royal Swedish Academy of Engineering Sciences. He received the 2021 IEEE PES Prabha S.Kundur Power System Dynamics and Control Award.

\end{IEEEbiography}
\end{document}